\documentclass[pra,reprint,amsmath,amssymb,amsthm,aps,superscriptaddress]{revtex4-1}

\usepackage{graphicx}% Include figure files
\usepackage{dcolumn}% Align table columns on decimal point
\usepackage{bm}% bold math
\usepackage{float}
\usepackage{appendix}

\usepackage{todonotes}
\usepackage{amsthm}
\usepackage{algorithm}
\usepackage{algpseudocode}
\usepackage{tabularx}
\usepackage{comment}
\usepackage[normalem]{ulem}
\usepackage[colorlinks, linkcolor=blue,anchorcolor=blue,citecolor=blue,urlcolor=blue]{hyperref}

\newtheorem{theorem}{Theorem}
\newtheorem{lemma}{Lemma}

\DeclareRobustCommand{\rchi}{{\mathpalette\irchi\relax}}
\newcommand{\irchi}[2]{\raisebox{\depth}{$#1\chi$}} % inner command, used by \rchi

\begin{document}

\title{Efficient algorithms to solve atom reconfiguration problems. II.~The assignment-rerouting-ordering (aro) algorithm}

\author{Remy El Sabeh}
\affiliation{Department of Computer Science, American University of Beirut, Lebanon.}

\author{Jessica Bohm}
\affiliation{Institute for Quantum Computing, University of Waterloo, Waterloo, Ontario N2L 6R2, Canada.}

\author{Zhiqian Ding}
\affiliation{Institute for Quantum Computing, University of Waterloo, Waterloo, Ontario N2L 6R2, Canada.}

\author{Stephanie Maaz}
\affiliation{David R. Cheriton School of Computer Science, University of Waterloo, Waterloo, Ontario N2L 3G1, Canada.}

\author{Naomi Nishimura}
%\email[]{nishi@uwaterloo.ca }
\affiliation{David R. Cheriton School of Computer Science, University of Waterloo, Waterloo, Ontario N2L 3G1, Canada.}

\author{Izzat El Hajj}
%\email[]{izzat.elhajj@aub.edu.lb}
\affiliation{Department of Computer Science, American University of Beirut, Lebanon.}

\author{Amer E. Mouawad}
%\email[]{aa368@aub.edu.lb}
\affiliation{Department of Computer Science, American University of Beirut, Lebanon.}
\affiliation{David R. Cheriton School of Computer Science, University of Waterloo, Waterloo, Ontario N2L 3G1, Canada.}
\affiliation{University of Bremen, Bremen, Germany}

\author{Alexandre Cooper}
\email[]{alexandre.cooper@uwaterloo.ca}
\affiliation{Institute for Quantum Computing, University of Waterloo, Waterloo, Ontario N2L 6R2, Canada.} 

\date{\today}

\begin{abstract}
Programmable arrays of optical traps enable the assembly of configurations of single atoms to perform controlled experiments on quantum many-body systems.
Finding the sequence of control operations to transform an arbitrary configuration of atoms into a predetermined one requires solving an atom reconfiguration problem quickly and efficiently.
A typical approach to solve atom reconfiguration problems is to use an assignment algorithm to determine which atoms to move to which traps.
This approach results in control protocols that exactly minimize the number of displacement operations; however, this approach does not optimize for the number of displaced atoms or the number of times each atom is displaced, resulting in unnecessary control operations that increase the execution time and failure rate of the control protocol.
In this work, we propose the assignment-rerouting-ordering (aro) algorithm to improve the performance of assignment-based algorithms in solving atom reconfiguration problems. 
The aro algorithm uses an assignment subroutine to minimize the total distance traveled by all atoms, a rerouting subroutine to reduce the number of displaced atoms, and an ordering subroutine to guarantee that each atom is displaced at most once.
The ordering subroutine relies on the existence of a partial ordering of moves that can be obtained using a polynomial-time algorithm that we introduce within the formal framework of graph theory.
We numerically quantify the performance of the aro algorithm in the presence and in the absence of loss, and show that it outperforms the exact, approximation, and heuristic algorithms that we use as benchmarks. 
Our results are useful for assembling large configurations of atoms with high success probability and fast preparation time, as well as for designing and benchmarking novel atom reconfiguration algorithms.
\end{abstract}

\maketitle

\section{Introduction}

Programmable arrays of optical traps~\cite{Tervonen1991, Prather1991, Curtis2002, Hossack2003} have recently emerged as effective tools for assembling configurations of single atoms and molecules with arbitrary spatial geometries~\cite{Bergamini2004, Lee2016, Kim2016, Endres2016, Barredo2016, Barredo2018, Anderegg2019}.
Supplemented with strong and tunable interactions like Rydberg-Rydberg interactions~\cite{Gallagher1994, Lukin2001}, these configurations realize large, coherent quantum many-body systems that act as versatile testbeds for quantum science and technology~\cite{Saffman2010,Browaeys2020,Morgado2021,Daley2022}.

An ongoing challenge is to assemble configurations of thousands of atoms with high success probability and fast preparation time. Addressing this challenge requires the design and implementation of improved algorithms to solve atom reconfiguration problems~\cite{Schymik2020, Cooper2024, Cimring2023}, which are hard combinatorial optimization problems that seek a sequence of control operations to prepare a given configuration of atoms from an arbitrary one. In the absence of atom loss, finding a control protocol that exactly minimizes the total number of displacement operations, without concern for any other performance metrics, can be done in polynomial time using \emph{assignment algorithms}, such as those based on the Hungarian algorithm~\cite{kuhn1955hungarian,edmonds1972theoretical,Lee2017}.

Relying on assignment algorithms to solve atom reconfiguration problems~\cite{Schymik2020,Cooper2024,Cimring2023}, however, suffers from two major drawbacks.
The first drawback is that these assignment-based  algorithms do not optimize for the number of displaced atoms; in fact, minimizing the total number of displaced atoms is an \textsf{NP}-complete problem, even on simple graphs and geometries such as grids, for which an approximate solution can be obtained in polynomial time using the \emph{Steiner tree 3-approximation algorithm} (3-approx)~\cite{Calinescu2007}.
If single atoms are displaced sequentially, then increasing the number of displaced atoms increases the number of transfer operations required to extract and implant the atoms from and into the array of optical traps, and thus increases the probability of losing them. 
For example, an algorithm might choose to displace $M$ atoms once instead of one atom $M$ times, which, although equivalent in terms of the total number of displacement operations, results in greater uncertainty about which atoms will be lost, complicating the problem of efficiently allocating surplus atoms to replace lost ones.
The second drawback is that the moves are executed in an arbitrary order, without taking into account the possibility of early moves obstructing later moves; the same atom might thus be displaced multiple times, further increasing the probability of losing it.

%%% \begin fig:intro
\begin{figure}[t]
\includegraphics[]{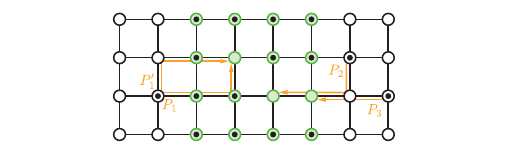}
\caption{
\textbf{The assignment-rerouting-ordering (aro) algorithm.}
The atom reconfiguration problem consists of finding a sequence of moves to transform an arbitrary configuration of atoms (black dots) contained in a static array of optical traps (circles) into a target configuration of atoms (shaded green disks). First, the aro algorithm uses the assignment subroutine to find a sequence of moves that minimizes the total distance traveled by all atoms, or, equivalently, the total number of displacement operations. Second, the rerouting subroutine attempts to update the path of each move to reduce the number of atoms displaced without increasing the total displacement distance, here choosing the path $P_1'$ over the path $P_1$. Third, the ordering subroutine finds a sequence of moves that prevents an atom from obstructing the path of another atom, here choosing to execute the move associated with $P_2$ before executing the move associated with $P_3$.}
\label{fig:intro}
\end{figure}
%%% \end fig:intro

In this paper, we propose the assignment-rerouting-ordering (aro) algorithm to overcome the aforementioned drawbacks.
Similar to typical assignment-based reconfiguration algorithms, the aro algorithm uses an assignment algorithm to determine which atoms to move to which traps; however, the aro algorithm then updates the resulting sequence of moves using two novel subroutines, which are our main contributions: a rerouting subroutine and an ordering subroutine.
The rerouting subroutine attempts to reroute the path of each move to reduce the number of displaced atoms, whereas the ordering subroutine orders the sequence of moves to guarantee that each atom is displaced at most once~(Fig.~\ref{fig:intro}). Both subroutines might thus modify the sequence of moves without increasing the total number of displacement operations. 
Hence, besides possibly achieving a reduction in the number of transfer operations, the aro algorithm exactly minimizes both the number of displacement operations and the number of transfer operations per displaced atom.
The resulting reduction in the total number of control operations allows the aro algorithm to outperform a typical assignment-based algorithm, the 3-approx algorithm, as well as our recently introduced redistribution-reconfiguration (red-rec) algorithm~\cite{Cimring2023}, at solving atom reconfiguration problems both in the absence and in the presence of loss.

More generally, we are motivated by the goal of improving the performance of atom reconfiguration algorithms ``from the ground up'' by building upon exact and approximation algorithms for which provable analytical guarantees exist, e.g., within the framework of combinatorial optimization and graph theory. 
This approach is complementary to operationally-driven approaches that build upon intuition and operational constraints to formulate heuristic algorithms~\cite{Barredo2016, Barredo2018, Schymik2020, Sheng2021, Mamee2021, Ebadi2021, Tao2022, Sheng2022, Cimring2023,Tian2023}. 

Exact and approximation algorithms can be used to improve operational performance as standalone algorithms or as subroutines within other heuristic algorithms, and they can also be used to provide performance bounds to benchmark new algorithms and identify ways to further improve them.
Moreover, these algorithms and the formal results that underpin them can support applications in other areas, e.g., robot motion planning~\cite{vanderBerg2009, Heuvel2013, Solovey2014, Demaine2019}.

For the purpose of this paper, we refer to ``assignment algorithms'' as algorithms that solve assignment problems and ``assignment-based reconfiguration algorithms'' or ``assignment-based algorithms'' as reconfiguration algorithms that solve atom reconfiguration problems using an assignment algorithm as a subroutine. When there is no risk of confusion, the term ``assignment'' or ``assignment algorithm'' might be used to denote a typical assignment-based reconfiguration algorithm that does not exploit the rerouting or ordering subroutines.

The rest of the paper is organized as follows.
We start by reviewing atom reconfiguration problems and describing our baseline assignment-based reconfiguration algorithm~(Sec.~\ref{sec:problem}), which is used as a point of reference when quantifying operational performance.
We then introduce the aro algorithm~(Sec.~\ref{sec:aro_algo}), describing in detail the rerouting subroutine~(Sec.~\ref{sec:rerouting}) and the ordering subroutine~(Sec.~\ref{sec:ordering}).
Next, we numerically benchmark the performance of the aro algorithm against exact, approximation, and heuristic algorithms, first, in the absence of loss~(Sec.~\ref{sec:benchmarking_lossless}), and then, in the presence of loss~(Sec.~\ref{sec:benchmarking_loss}). 
We provide supporting proofs and technical details about the various subroutines, including runtime analysis and proofs of correctness, in the appendices~(App.~\ref{app:assignment}~--~\ref{sec:batching}).

\section{Atom reconfiguration problems}\label{sec:problem}

An atom reconfiguration problem~\cite{Schymik2020, Cooper2024, Cimring2023} seeks a control protocol, $\mathcal{R}:\mathcal{C}_0\mapsto\mathcal{C}_T$, to transform an arbitrary configuration $\mathcal{C}_0$ of $N_a^0$ atoms into a given target configuration $\mathcal{C}_T$ of $N_a^T$ atoms.
A configuration of atoms is contained in an array of optical traps, $\mathcal{A}(V)$, defined by its spatial arrangement or \emph{geometry}, $V = \{\vec{v}_j~|~\vec{v}_j = (v_{j_x}, v_{j_y}) \in \mathbb{R}^2, 1 \leq j \leq N_t\}$, where $N_t$ is the number of traps in the optical trap array. We later choose the geometry to be a square lattice (a grid) of $N_t=N_t^x\times N_t^y$ traps in the plane where $v_{j_x}=x_0+j_x\delta x$, $v_{j_y}=y_0+j_y\delta y$ for $0\leq j_x\leq N_t^x - 1$, $0 \leq j_y \leq N_t^y-1$, $(x_0, y_0)$ is the origin of the array, and $\delta x,~\delta y$ are the lattice spacing constants.

The control protocol is composed of a sequence of extraction-displacement-implantation (EDI) cycles that extract, displace, and implant a single atom from one static trap to another using a secondary array of dynamic traps. This operation can also be done on multiple atoms simultaneously. These EDI cycles are composed of a sequence of elementary control operations that include elementary transfer operations, which extract (implant) an atom from (into) a static trap into (from) a dynamic trap, and elementary displacement operations, which displace a dynamic trap containing an atom from one static trap to another by an elementary displacement step $\delta x$ or $\delta y$. In order to account for the probability of losing an atom even if a trap stays idle, we also include no-op operations that leave some traps unchanged, while transfer or displacement operations are performed on other traps.
%The various transfer and displacement operations might be executed in series or in parallel depending on the operational capability of the control hardware. Here, besides the use of a batching subroutine~(Sec.~\ref{sec:batching}), we focus on sequential reconfiguration protocols that displace or transfer one atom per control step. We also restrict displacement operations along trajectories defined over the trap array, e.g., instead of between traps, which is a valid approach when atoms are indistinguishable and interchangeable.

In an operational setting, the initial configuration of atoms is obtained by randomly loading a single atom into every trap of the trap array with a probability given by the \emph{loading efficiency} $\epsilon$. Given the initial and target configurations of atoms, the reconfiguration problem is then solved, the control protocol is executed, and a measurement is performed to check whether or not the updated configuration of atoms contains the target configuration. In the presence of loss, the atom reconfiguration problem might have to be solved multiple times through multiple reconfiguration cycles until the target configuration is reached (success) or is no longer reachable (failure). The same algorithm is used independently of the initial configuration, that is, we do not consider adaptive algorithms whose behavior changes depending on the measured configuration, nor do we consider protocols that rely on mid-cycle measurements.

\subsection{Atom reconfiguration problems on graphs}
Atom reconfiguration problems can be viewed as reconfiguration problems on graphs~\cite{Heuvel2013,Nishimura2018,Bousquet2022, Calinescu2007,Ito2011,Cooper2024}. A configuration of indistinguishable atoms trapped in an array of optical traps is represented as a collection of tokens placed on a subset of the vertices of a graph, $G$, where $V(G)$ and $E(G)$ denote the vertex set and edge set of $G$, respectively, with $|V(G)| = n$ and $|E(G)| = m$. We directly use $n$ and $m$ to refer to the number of vertices and the number of edges, respectively, when the graph $G$ is clear from the context. We assume that each graph is finite, simple, connected, undirected, and edge-weighted (we refer to Diestel's textbook~\cite{diestel} for standard graph terminology). We use $w: E(G) \rightarrow \mathbb{N}^{+}$ to denote the edge-weight function which implies that $w(e)$ is positive for all $e = \{u,v\} \in E(G)$.

Although our main result, Theorem~\ref{thm:broken_path_system}, is valid for arbitrary (positive edge-weighted) graphs, we focus on (unweighted) \emph{grid graphs}. Specifically, we focus on the \emph{$(p \times q)$-grid graph}, which is a graph of $p \times q$ vertices with vertex set $\{(x, y) \mid x \in \{0, 1, \cdots, p - 1\},~y \in \{0, 1, \cdots, q - 1 \}\}$ for $p,q\in\mathbb{N}^{+}$. We denote the \emph{width} $p$ of a grid graph $G$ by $W_G$, and its \emph{height} $q$ by $H_G$ (we drop the subscript $G$ when the graph we are referring to is clear from the context). Two vertices $v=(x, y)$ and $v'=(x', y')$ ($v \neq v'$) are adjacent, and thus connected by an edge, if and only if $|x - x'| + |y - y'| \leq 1$. Note that $m = \mathcal{O}(n)$ whenever $G$ is a planar graph, as is the case for grid graphs.

In addition to the graph $G$, the atom reconfiguration problem requires definitions of the initial (source) and desired (target) configurations of atoms. The traps containing the atoms in the source and target configurations are identified as subsets of vertices $S \subseteq V(G)$ and $T \subseteq V(G)$, respectively. We assume that $|S| \geq |T|$ since otherwise the problem does not have a solution. Note that $S$ and $T$ need not be disjoint. Each vertex in $S$ has a token on it and the problem is to move the tokens on some $S^\star \subseteq S$ such that all vertices of $T$ eventually contain tokens.

Here, a \emph{move} of token $\tau_i$ ($1 \leq i \leq |S|$) from vertex $u$ to vertex $v$, which is equivalent to a sequence of elementary displacement operations, is \emph{allowed} or \emph{unobstructed} whenever $\tau_i$ is on $u$ and the path $P$ (defined formally in the next section) from $u$ to $v$ in $G$ associated with it is free of tokens (except for $\tau_i$); otherwise, we say that the move is \emph{obstructed} and call each token $\tau_j$ ($j \neq i$) on $P$ an \emph{obstructing token}.
If we attempt to move a token along a path that is not free of tokens, then we say that this move causes a \emph{collision}. Because a collision induces the loss of the colliding atoms, moves that cause collisions are replaced by sequences of moves that do not cause collisions.
Indeed, if the move of token $\tau_i$ from $u$ to $v$ is obstructed, then, assuming $v$ is free of tokens, we can always reduce this move to a sequence of unobstructed moves by replacing the move by a sequence of moves involving the obstructing tokens, i.e., solving the \emph{obstruction problem} (which we solve using a slightly different procedure described in App.~\ref{sec:obstruction_solver}).
A solution to an atom reconfiguration problem is thus a sequence of unobstructed moves, each of which displaces a token from a vertex with a token to a vertex without a token along a path that is free of obstructing tokens.  

\subsection{Path systems as solutions to atom reconfiguration problems on graphs}

Our reconfiguration algorithm constructs a \emph{valid path system} (defined below), updates it, and then finds a sequence of unobstructed moves to execute along every path.

We define a \emph{path} in a graph $G$ as a walk whose sequence of vertices comprises distinct vertices. We define a \emph{walk} (of length $\ell$) in $G$ as a sequence of vertices in $V(G)$, $(v_0, \ldots, v_\ell)$, such that $\{v_i,v_{i + 1}\} \in E(G)$ for all $i \in \{0, \ldots, \ell-1\}$, where $\{v_1, v_2, \ldots, v_{\ell-1}\}$ are the \emph{internal} vertices of the walk. We define a \emph{cycle} in $G$ as a walk of length $\ell \geq 3$ that starts and ends on the same vertex, $v_0 = v_\ell$, and whose internal vertices form a path. 
The weight of a path is given by the distance between its first ($v_0$) and last vertex ($v_\ell$), where the \emph{distance} between $u$ and $v$ in $G$ is the weight of a \emph{shortest path} between $u$ and $v$, computed as the sum of the weights of the edges connecting the vertices of a shortest path $P$ between $u$ and $v$, $w(P) = \sum_{e \in P}{w(e)}$. When the graph is unweighted, or each of its edges has a weight of one, in which case the graph is said to be \emph{uniformly-weighted}, then the distance between $u$ and $v$ corresponds to the minimum number of edges required to get from $u$ to $v$ in $G$. 

A \emph{path system} $\mathcal{P}$ in $G$ is a collection of paths, $\mathcal{P} = \{P_1, P_2, \ldots, P_k\}$, in which each path $P_i \in \mathcal{P}$ for $i \in \mathbb{N}^{+}([1,k])$ is a path from $v_{s_i}$ (source vertex) to $v_{t_i}$ (target vertex), which we denote by $\{v_{s_i}, v_1, v_2, \ldots, v_{t_i}\}$ (single-vertex paths with $v_{s_i} = v_{t_i}$ are also allowed).
We define a \emph{doubly-labeled} vertex as a vertex that is both a source vertex and a target vertex, whether within the same path or in two different paths.
The weight of a path system is given by the sum of the weights of its paths, $w(\mathcal{P}) = \sum_{P \in \mathcal{P}}{w(P)}$. 
Each source vertex $v_{s_i} \in V(P_i)$ associated with a path $P_i\in\mathcal{P}$ contains a token, i.e., $v_{s_i} \in S$; the other vertices in $P_i$ may or may not contain tokens.
A token $\tau_i$ is said to be \emph{isolated} in a path system $\mathcal{P}$ whenever there exists a single-vertex path $P = \{v\} \in \mathcal{P}$ such that the token $\tau_i$ is on the vertex $v$ and no other path in $\mathcal{P}$ contains vertex $v$.

We say that a move (of a non-isolated token) associated with path $P_i \in \mathcal{P}$ is \emph{executable} whenever the target vertex $v_{t_i}$ does not contain a token;
an unobstructed move is trivially executable, whereas an obstructed move can always be reduced to a sequence of unobstructed moves, assuming $v_{t_i}$ contains no token. A path system $\mathcal{P}$ is said to be \emph{valid (for $T$)} or \emph{$T$-valid} whenever there exists some ordering of the moves that makes them executable, and executing all the moves associated with $\mathcal{P}$ results in each vertex in $T$ having a token on it (with the exception of isolated tokens, which need not move). Clearly, in a valid path system, all source vertices are distinct and all target vertices are distinct, although some source vertices can be the same as some target vertices. 
We note that for any valid path system, we can always find an executable move, unless the problem has already been solved with all vertices in $T$ occupied by tokens. 
We also note that, before any move is executed, whenever we have a token on some target or internal vertex, then there must exist a path for which this token is on the source vertex. 

As described in the next subsection~(Sec.~\ref{sec:assignment}), a typical assignment-based reconfiguration algorithm solves an assignment problem to compute a valid path system of minimum weight, i.e., in which each path is one of the many possible shortest paths between its source vertex and its target vertex, chosen arbitrarily among the set of all possible shortest paths. The obstruction problem is then solved to find a sequence of unobstructed moves on each path. Our proposed aro algorithm~(Sec.~\ref{sec:aro_algo}) solves an assignment problem to compute a valid path system, runs the rerouting subroutine, and then runs the ordering subroutine, which yields an ordering of the paths such that the resulting sequence of associated moves guarantees that each move is unobstructed at the time of its execution.
 
\subsection{Baseline reconfiguration algorithm}\label{sec:assignment}

A typical approach to solve an atom reconfiguration problem is to map it onto an assignment problem, which can be solved in polynomial time using an assignment algorithm, such as one based on the Hungarian algorithm~\cite{kuhn1955hungarian,edmonds1972theoretical,Lee2017}. A typical assignment-based reconfiguration algorithm computes a \emph{(valid) distance-minimizing path system}, which is a valid path system, $\mathcal{P}$, whose weight $w(\mathcal{P})$ is minimized, i.e., the resulting control protocol exactly minimizes the total number of displacement operations performed on all atoms. 

\begin{algorithm}[H]
\caption{-- A baseline reconfiguration algorithm}
\label{alg:assignment}
\begin{algorithmic}[1]
\Require A static trap array, $A$, represented as a  positive edge-weighted graph $G=(V,E)$ with $\sum_{e \in E(G)}{w(e)} = \mathcal{O}(n^c)$ for some positive integer $c$; an initial configuration of atoms, $C_0$, represented as a set of source vertices, $S \subseteq V(G)$; and a target configuration of atoms, $C_T$, represented as a set of target vertices, $T \subseteq V(G)$.
\State Compute the distance and a shortest path between all pairs of vertices of $S$ and $T$ by solving the all-pairs shortest path (APSP) problem~($\mathcal{O}(n^3)$, $\mathcal{O}(n^2)$ on uniformly-weighted grid graphs).
\State Compute a distance-minimizing path system, matching every vertex in $T$ to a distinct vertex in $S$ and forming a shortest path between them by solving the assignment problem~($\mathcal{O}(n^3)$).
\State (optional) Using the isolation subroutine, modify the path system to (locally) isolate tokens that do not need to be displaced~($\mathcal{O}(n^5)$).
\State Using the obstruction solver subroutine, find a sequence of unobstructed moves~($\mathcal{O}(n^2)$).
\State (optional) Using the batching subroutine, batch moves to perform control operations on multiple atoms in parallel~($\mathcal{O}(n^3)$).
\end{algorithmic}
\end{algorithm}

To benchmark the performance of our proposed aro algorithm, we use a \emph{baseline reconfiguration algorithm}~(Alg.~\ref{alg:assignment}), which is a slightly modified version of a typical assignment-based algorithm that relies primarily on solving an assignment problem.
This baseline algorithm solves atom reconfiguration problems in five steps, two of which are optional.
In the first and second steps, the \emph{assignment subroutine} computes a valid distance-minimizing path system by solving the all-pairs shortest path (APSP) problem, followed by the assignment problem.
In the third optional step, the \emph{isolation subroutine} isolates a maximal subset of tokens (not contained in a larger subset) found on doubly-labeled vertices.
The purpose of the isolation subroutine is to (attempt to) decrease the number of atoms that have to move by greedily fixing in place some atoms whose deletion from the graph (along with the vertices on which they lie) does not increase the weight of a distance-minimizing path system recomputed in the resulting graph.
We find a maximal subset without the guarantee that it is a maximum subset (the largest subset in the whole graph), because finding a maximum subset is equivalent to the \textsf{NP}-complete problem of minimizing the number of tokens that move, which remains \textsf{NP}-complete even on grids~\cite{Calinescu2007}. 
In the fourth step, the \emph{obstruction solver subroutine} computes a sequence of unobstructed moves associated with the path system. 
In the fifth optional step, the \emph{batching subroutine} combines some of the moves to simultaneously displace multiple atoms in parallel.
We describe each of these subroutines in more detail in App.~\ref{app:assignment}, App.~\ref{app:isolation}, App.~\ref{sec:obstruction_solver}, and App.~\ref{sec:batching}, respectively.

\section{The assignment-rerouting-ordering algorithm}~\label{sec:aro_algo}

To improve the performance of assignment-based reconfiguration algorithms, we propose the assignment-rerouting-ordering (aro) algorithm, which exploits a rerouting subroutine (Sec.~\ref{sec:rerouting}) and an ordering subroutine (Sec.~\ref{sec:ordering}).
The aro algorithm performs fewer transfer operations than the baseline reconfiguration algorithm while still minimizing the total number of displacement operations, thereby strictly improving overall performance in the absence and in the presence of loss. 

\begin{algorithm}[H]
\caption{-- The aro algorithm}\label{alg:aro}
\begin{algorithmic}[1]
\Require A static trap array, $A$, represented as a positive edge-weighted graph $G=(V,E)$ with $\sum_{e \in E(G)}{w(e)} = \mathcal{O}(n^c)$ for some positive integer $c$; an initial configuration of atoms, $C_0$, represented as a set of source vertices, $S \subseteq V(G)$; and a target configuration of atoms, $C_T$, represented as a set of target vertices, $T \subseteq V(G)$.
\State Compute the distance and a shortest path between all pairs of vertices of $S$ and $T$ by solving the all-pairs shortest path (APSP) problem~($\mathcal{O}(n^3)$, $\mathcal{O}(n^2)$ on uniformly-weighted grid graphs).
\State Compute a distance-minimizing path system, matching every vertex in $T$ to a distinct vertex in $S$ and forming a shortest path between them by solving the assignment problem~($\mathcal{O}(n^3)$).
\State (optional) Using the isolation subroutine, modify the path system to (locally) isolate tokens that do not need to be displaced~($\mathcal{O}(n^5)$).
\State Using the rerouting subroutine, reroute the paths in the path system to increase the number of isolated tokens in the path system~($\mathcal{O}(n^5)$ assuming the distance-preserving rerouting subroutine is used, $\mathcal{O}(n^3)$ on uniformly-weighted grid graphs).
\State Using the ordering subroutine, order the paths in the path system, which breaks cycles if they exist~($\mathcal{O}(n^{c+6}m)$, $\mathcal{O}(n^8)$ on uniformly-weighted grid graphs).
\State Using the obstruction solver subroutine, find a sequence of unobstructed moves~($\mathcal{O}(n^2)$).
\State (optional) Using the batching subroutine, batch moves to perform control operations on multiple atoms in parallel~($\mathcal{O}(n^3)$).
\end{algorithmic}
\end{algorithm}

The aro algorithm~(Alg.~\ref{alg:aro}) solves atom reconfiguration problems in seven steps, five of which are imported from the baseline algorithm. In particular, Steps 1, 2, 3, 6 and 7 are standard subroutines that are shared with the baseline algorithm, and Steps 4, 5 are our main contributions.
In the first three steps, similarly to the baseline algorithm, the assignment subroutine computes a valid distance-minimizing path system by solving the all-pairs shortest path (APSP) and assignment problems, and optionally isolates a maximal subset of tokens located on doubly-labeled vertices by using the isolation subroutine. Next, instead of directly computing the sequence of moves to execute as in the baseline algorithm, the aro algorithm seeks to further update the path system. 
In the fourth step, the \emph{rerouting subroutine}~(Sec.~\ref{sec:rerouting}), which is a heuristic, seeks to reroute each path in the path system in an attempt to reduce the number of displaced atoms. 
In the fifth step, the \emph{ordering subroutine}~(Sec.~\ref{sec:ordering}) constructs an ordered path system that admits an ordering of its paths, guaranteeing that each atom moves at most once.
In the sixth step, similarly to the baseline algorithm, the \emph{obstruction solver subroutine} computes a sequence of unobstructed moves associated with the path system; because the paths are ordered, the move associated with each path is unobstructed, and solving the obstruction problem is trivial.
In the seventh (optional) step, the batching subroutine combines some of the moves to simultaneously displace multiple atoms in parallel.

Our current implementation of the aro algorithm has been designed to work with general edge-weighted graphs and runs in time $\mathcal{O}(n^{8})$ on uniformly-weighted grid graphs.
The correctness of the algorithm follows from Theorem~\ref{thm:ops} (App.~\ref{app:ops}), Lemma~\ref{lem-greedy-isolate} (App.~\ref{app:isolation}), and Lemma~\ref{lem-dp} (App.~\ref{app:rerouting}). Theorem \ref{thm:broken_path_system} summarizes all the aforementioned results. We note that our running time estimates on both general edge-weighted graphs and uniformly-weighted grid graphs are not optimized. In particular, we conjecture the existence of an implementation of the aro algorithm restricted to uniformly-weighted grid graphs that runs in $\mathcal{O}(n^4)$, which we believe follows almost immediately from Theorem~\ref{thm:broken_path_system}. We chose not to include the details of said implementation because it only works on uniformly-weighted grid graphs and the running time of the algorithm has no effect on the operational performance. Also, we believe that it can be optimized further, and our focus in this paper is primarily on proving the existence of a polynomial-time ordering subroutine.

\subsection{Rerouting subroutine}\label{sec:rerouting}
The baseline reconfiguration algorithm returns a path system that minimizes the total number of displacement operations, without considering the total number of displaced atoms or the total number of transfer operations. 
Because the problem of minimizing the number of transfer operations is an \textsf{NP}-complete problem~\cite{Calinescu2007} (even on grids), and finding a control protocol that simultaneously minimizes both displacement and transfer operations is impossible for some instances~(Fig.~\ref{fig:fig1a}c), we must resort to using heuristics that attempt to reduce the number of atoms that are displaced.

To reduce the number of displaced atoms while preserving the number of displacement operations, we rely on the \emph{distance-preserving rerouting subroutine}, which attempts to replace each path in the path system with another path of the same weight that contains fewer vertices occupied by atoms.
We refer to \emph{rerouting a path} as updating its sequence of internal vertices while preserving its source and target vertices.
An example of rerouting a path can be found in Fig.~\ref{fig:intro} where $P_1$ is rerouted to $P_1'$.
The intent behind the usage of the rerouting subroutine is to attempt to increase the number of isolated tokens, i.e., tokens that do not have to move.
This version of rerouting was designed specifically for uniformly-weighted grid graphs, but it can easily be generalized.

The subroutine proceeds by iterating over every path in the path system, and, for every path, attempting to reroute it, while keeping the rest of the path system unchanged, in a way that maximizes token isolation while preserving the weight of the path. If the original path is a straight line, then there is nothing to do, as the path cannot be rerouted without increasing its weight. Otherwise, suppose that the source vertex $v_s$ of the path is $(x_1, y_1)$ and the target vertex $v_t$ of the same path is $(x_2, y_2)$, and that, without loss of generality, $x_1 < x_2$ and $y_1 < y_2$. Given $W = |x_1 - x_2|$ and $H = |y_1 - y_2|$, there are a total of ${(W + H)!}/{H!W!}$ shortest paths between vertex $(x_1, y_1)$ and vertex $(x_2, y_2)$.
Using a brute-force approach for every path is inefficient, as the number of rerouted paths to consider for every path is exponential in the Manhattan distance between the source vertex and the target vertex of the path. To avoid an exhaustive search and to speed up computation, we exploit dynamic programming~(App.~\ref{app:rerouting}); exhaustively rerouting paths to increase the number of isolated tokens can then be performed in $\mathcal{O}(n^{5})$ time on general positive edge-weighted graphs and in $\mathcal{O}(n^{3})$ time on uniformly-weighted grid graphs~(Lemma~\ref{distance_preserving_running_time}).

A possible extension of the rerouting subroutine that we have developed, but whose performance we have not quantified, is to search for paths that might not necessarily preserve the minimum total displacement distance or minimum number of displacement operations. The \emph{distance-increasing rerouting subroutine}~(see App.~\ref{app:distance-increasing-rerouting} for a detailed presentation) trades off an increase in displacement operations for a decrease in transfer operations. This subroutine runs in $\mathcal{O}(n^{7})$ time on uniformly-weighted grid graphs  (Lemma~\ref{bounded_termination_increasing}).
We note that the aro algorithm would still be correct and would have the same asymptotic running time if we were to replace the distance-preserving rerouting subroutine with the distance-increasing rerouting subroutine. 

\subsection{Ordering subroutine}\label{sec:ordering}
The \emph{ordering subroutine} constructs an ordered path system that admits a (partial) ordering of its paths, so that the moves associated with the paths are unobstructed at the time of their execution, i.e., atoms displaced in preceding moves do not obstruct the displacement of atoms in succeeding moves and atoms obstructing certain paths are displaced before the atoms on the paths that they obstruct.
The ordering is obtained by finding the partial ordering of the vertices of the \textit{dependency graph}, which is a graph where each path is represented by a vertex, and where each \emph{dependency} of a path $P_i$ on a path $P_j$ is represented by a directed edge from the vertex representing path $P_j$ to the vertex representing path $P_i$. $P_i$ is said to \emph{depend on} $P_j$ if $v_{s_j}$ is an internal vertex in (or the target vertex of) $P_i$, or if $v_{t_i}$ is an internal vertex in (or the source vertex of) $P_j$.
For example, in a partial ordering of the paths in the path system in Fig.~\ref{fig:intro}, $P_2$ would precede $P_3$, as executing the move associated with $P_3$ before executing the move associated with $P_2$ would obstruct the move associated with $P_2$. Such partial orderings are easy to compute in \emph{cycle-free} path systems, which are path systems that not induce any cycle, i.e., that induce a cycle-free graph, or a forest.

This ordering of moves guarantees that each displaced atom undergoes exactly one EDI cycle, thereby restricting the number of transfer operations per displaced atom to its strict minimum of two (one extraction operation and one implantation operation per EDI cycle).
The existence of a polynomial-time procedure to transform any (valid) path system into a \emph{(valid) ordered path system}, in which the paths are ordered such that executing the moves associated with each path displaces every atom at most once, follows from Theorem~\ref{thm:broken_path_system}.

\begin{theorem}\label{thm:broken_path_system}
A valid path system $\mathcal{P}$ in a positive edge-weighted graph $G$ can always be transformed in polynomial time into a valid cycle-free path system $\mathcal{P'}$ such that $w(\mathcal{P'})\leq w(\mathcal{P})$. Moreover, the dependency graph associated with the path system $\mathcal{P'}$ is a directed acyclic graph (DAG), which admits a partial ordering of its vertices, implying a partial ordering of the corresponding moves. 
\end{theorem}

In simple terms, Theorem~\ref{thm:broken_path_system} states that a path system can always be transformed (in polynomial time) into an ordered path system that admits an ordering of its paths such that executing the moves associated with the paths is guaranteed not to cause any collisions.
This theorem is valid for any arbitrary path system defined over any arbitrary (edge-weighted) graph. 

Theorem~\ref{thm:broken_path_system} applies to general (positive) edge-weighted graphs, even when the path system is not distance-minimizing, e.g., is obtained from implementing the distance-increasing rerouting subroutine. Its correctness directly follows from the correctness of the ordering subroutine~(see Alg.~\ref{alg:ordering}), which efficiently constructs a cycle-free path system and finds the ordering of the paths within it. We now describe the four steps of the ordering subroutine, providing formal proofs of supporting lemmas and theorems in App.~\ref{app:ordering}.

\begin{algorithm}[H]
\caption{-- The ordering subroutine}\label{alg:ordering}
\begin{algorithmic}[1]
\Require A valid path system, $\mathcal{P}$, defined on a positive edge-weighted graph $G$ with $\sum_{e \in E(G)}{w(e)} = \mathcal{O}(n^c)$ for some positive integer $c$.
\State Merge path system ($\mathcal{O}(n^{c+4}m^2)$, $\mathcal{O}(n^7)$ on uniformly-weighted grid graphs).
\State Unwrap path system ($\mathcal{O}(n^2m)$, $\mathcal{O}(n^{3})$ on uniformly-weighted grid graphs).
\State Detect/break cycles in path system ($\mathcal{O}(n^{c + 6}m)$, $\mathcal{O}(n^{8})$ on uniformly-weighted grid graphs).
\State Order path system ($\mathcal{O}(n^3)$). 
\end{algorithmic}
\end{algorithm}

%Step 1 -- Merging
First, the \emph{merging step} converts a path system into a (non-unique) \emph{merged path system} (\emph{MPS}). 
A merged path system is a path system such that no two paths intersect more than once, with the intersecting sections of the two paths possibly involving more than one vertex (all vertices in the intersection being consecutive in the vertex sequence of both paths). The merging operation does not increase the total weight of the path system, but it can decrease it; if the path system is generated by the Hungarian algorithm, then the total weight of the path system is already minimized.
A valid merged path system can be computed in time $\mathcal{O}(n^{c+4}m^2)$ for a graph $G$ where $\sum_{e \in E(G)}{w(e)} = \mathcal{O}(n^{c})$ for some positive integer~$c$ (see Lemma~\ref{lemma:mps} in App.~\ref{app:mps}).

%Step 2 -- Untangling
Second, the \emph{unwrapping step} converts an MPS into an \emph{unwrapped path system} (\emph{UPS}) by recomputing tangled paths. An unwrapped path system is a MPS such that no two paths within it are tangled. 
Two paths $P_i$ and $P_j$ are \emph{tangled} if $P_i$ wraps or is wrapped by $P_j$, where a path $P_i$ is said to be \emph{wrapped} in another path $P_j$ if it is entirely contained in it; if $P_i$ is wrapped by $P_j$, then $P_j$ wraps $P_i$. Unwrapping the paths that a path $P_j$ wraps is performed by sorting their respective source traps and target traps separately based on their order of occurrence within $P_j$, and assigning every source trap to the target trap of the same order. 
A valid UPS can be obtained from a valid MPS in time $\mathcal{O}(n^2m)$~(see Lemma~\ref{lemma:ups} in App.~\ref{app:ups}).

%Step 3 -- Cycle-breaking step
Third, the \emph{cycle-breaking step} converts an UPS into a \emph{cycle-free path system} (CPS).
The cycle-breaking step modifies the path system such that the graph induced on the modified path system is a cycle-free graph (a forest). 
Even though a graph can have exponentially many cycles, we prove that the cycle-breaking step can be executed in polynomial time, i.e., a valid CPS can be obtained from a valid UPS in time $\mathcal{O}(n^{c + 6}m)$~(see Theorem~\ref{lemma:cycle_breaking_termination} in App.~\ref{app:cps}). This result relies on the existence of ``special cycles'' in a graph with cycles~(App.~\ref{app:special-cycle}), which can be found using a procedure that we provide in App.~\ref{app:find-cycles}. Once a special cycle has been found, the set of paths that induce it can be found and updated
to break the special cycle~(see App.~\ref{app:break-cycles}), and the process is repeated on the resulting path system. This subroutine terminates in polynomial time and produces a cycle-free path system~(see Lemma~\ref{lemma:cycle_breaking_termination} in App.~\ref{app:termination}). %Our approach to modify the paths that induce a special cycle to guarantee that they no longer induce a cycle without increasing the total length of the path system in App.~\ref{sec:justification}.
The reason for using the (time-consuming) cycle-breaking procedure to break cycles instead of, e.g., computing minimum spanning trees (MSTs) using Theorem~{2.1} of Călinescu \emph{et al.}~\cite{Calinescu2007}, is that computing an MST to break cycles might in fact increase the total weight of the path system~(see Fig.~\ref{fig:mst_counterexample} for an example on a weighted graph). 
%However, we note that our current implementation largely follows the proof of existence of cycle-free path systems and more efficient implementations $(\mathcal{O}(n^4))$ on grids are possible~\cite{}. For instance, a different algorithm could iterate over edges of cycles formed by the paths in the path system and attempt to delete them one by one; the deletion of an edge is made permanent whenever we can still find a valid path system of the appropriate weight in the graph minus the edge. The procedure is then repeated on the new graph, i.e, the graph minus the edge, until we obtain a cycle-free path system. Proving the correctness of such an algorithm requires proving Theorem~\ref{thm:broken_path_system} and it is therefore important to note that our presentation is oriented towards simplifying the proofs of correctness rather than optimizing the worst-case asymptotic running time of the aro algorithm.

%Example of a path system that induces a cycle that cannot be broken via computing a MST without increasing total path system weight
\begin{figure}[t]
\includegraphics[]{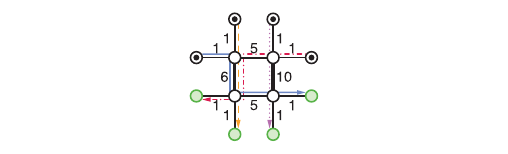}
\caption{
\label{fig:mst_counterexample}
\textbf{Justification for the cycle-breaking procedure.} 
Example of a path system defined on a weighted graph that induces a cycle that cannot be broken by computing a minimum spanning tree (MST) without increasing the total weight of the path system. The initial path system has a total weight equal to $46$ (the red path and the blue path each has a weight of 13, the orange path has a weight of 8 and the purple path has a weight of 12). The MST of the graph induced on the path system includes all the edges of the graph except for the edge of weight $10$. The path system generated by computing all-pairs shortest paths on the MST and then computing a source-target matching has a total weight of $52 = 1 \times 8 + 5 \times 4 + 6 \times 4$, irrespective of the specific choice of the matching (among $4! = 24$ possibilities), as the edges of weight $5$ will each be in two paths and the edge of weight $6$ will be in $4$ paths. A similar example can be constructed for the case of uniformly-weighted graphs, where deleting some edge might increase the weight of the path system.
}
\end{figure} 

%If we restrict ourselves to uniformly-weighted graphs, then it may be possible to replace the current cycle-breaking procedure with a more efficient one to obtain an ordered path system, e.g., by carefully using MSTs. 
%We note that Thm.~\ref{thm:broken_path_system} applies to general (positive) edge-weighted graphs, even when the path system is not distance-minimizing, e.g., is obtained from implementing the distance-increasing rerouting subroutine.

Fourth, the \emph{ordering step} constructs an \emph{ordered path system} (OPS) by ordering the moves associated with a CPS, e.g., by constructing the DAG associated with the path system. This step can be performed in time $\mathcal{O}(n^{3})$~(see Theorem~\ref{thm:ops} in 
App.~\ref{app:ops}). The moves associated with the OPS can be trivially computed using the obstruction solver subroutine~(see Sec.~\ref{sec:obstruction_solver}).

Having described in detail the baseline reconfiguration algorithm and the aro algorithm, we now proceed to numerically quantify their operational performance in the absence of loss~(see Sec.~\ref{sec:benchmarking_lossless}) and in the presence of loss~(see Sec.~\ref{sec:benchmarking_loss}).

\section{Quantifying performance in the absence of loss}~\label{sec:benchmarking_lossless}

We numerically quantify the performance of the aro algorithm in the absence of loss. We choose the performance metrics to be the total number of displacement operations, $N_\nu$, the total number of transfer operations, $N_\alpha$, and the total number of control operations, $N_T=N_\alpha+N_\nu$. The values computed for these performance metrics can be compared to the values obtained using the baseline reconfiguration algorithm and the 3-approx algorithm; the total number of displacement operations is minimized by using the baseline reconfiguration algorithm, whereas the total number of transfer operations is bounded by at most 3~times its optimal value by using the 3-approx algorithm.
These performance metrics directly correlate with the operational performance obtained in the presence of loss~(see Sec.~\ref{sec:benchmarking_loss}), quantified in terms of the mean success probability; reducing the total number of displacement and transfer operations results in fewer atoms lost, whereas reducing the number of displaced atoms concentrates the loss probability in as few atoms as possible, simplifying the problem of filling up empty target traps in subsequent reconfiguration cycles.  

Although our results are valid for any arbitrary geometry, and more generally for atom reconfiguration problems defined on arbitrary graphs, we focus our analysis on the problem of preparing compact-centered configurations of $N_a^T=\sqrt{N_a^T}\times\sqrt{N_a^T}$ atoms in rectangular-shaped square-lattice arrays of $N_{t}=\sqrt{N_a^T}\times\eta\sqrt{N_a^T}=\eta N_a^T$ static traps, where $\eta=N_{t}/N_a^T$ is the \emph{trap overhead factor}, which quantifies the overhead in the number of optical traps needed to achieve a desired configuration size.
In the absence of loss, the overhead factor is typically chosen based on the desired baseline success probability, which depends on the probability of loading at least as many atoms as needed to satisfy the target configuration, and thus on the loading efficiency, $\epsilon$. 
In the presence of loss, the overhead factor is typically a non-linear function of the configuration size that is chosen based on the desired mean success probability~\cite{Cimring2023}.

%%% \begin fig:lossless_rerouting_ordering
\begin{figure}[t]
\includegraphics[]{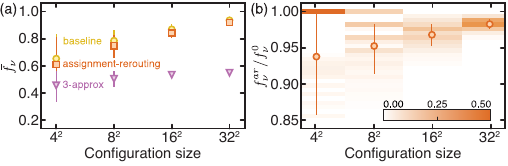}
\caption{
\label{fig:lossless_rerouting_ordering}
\textbf{Reducing the fraction of displaced atoms using the rerouting subroutine in the absence of loss.}
(a)~Mean fraction of displaced atoms, $\bar{f}_\nu=\langle N_a^\nu/N_a^0\rangle$, computed for various configuration sizes for the baseline (yellow circle), assignment-rerouting (orange square), and 3-approx (purple inverted triangles) algorithms. 
(b)~Distribution of the relative fraction of displaced atoms and its mean value for various configuration sizes computed as the ratio of the fraction of displaced atoms for the assignment-rerouting algorithm and the baseline reconfiguration algorithm.
}
\end{figure}
%%% \end fig:lossless_rerouting_ordering

For each target configuration size, we sample over a thousand initial configurations of atoms by distributing $N_a^0$ atoms at random over $N_{t}$ traps. We then count the number of transfer and displacement operations for each displaced atom within each realization and compute the ensemble average over the distribution of initial configurations. The number of atoms in the initial configuration satisfies a binomial distribution, $N_a^0 \sim$ \text{Bino}($N_{t},\epsilon$), where the loading efficiency is conservatively chosen to be $\epsilon=0.5$ and the trap overhead factor is chosen to be $\eta=1/\epsilon=2$. As computed from the cumulative distribution function of the binomial distribution, the baseline success probability is thus $\bar{p}=0.5$, i.e., half the initial configurations contain enough atoms to satisfy the target configuration; however, we restrict our analysis to successful reconfiguration protocols with $N_a^0\geq N_a^T$. 

%\subsection{Rerouting}
We first compute the reduction in the number of displaced atoms, $N_a^\nu$, or, equivalently, the fraction of displaced atoms, $f_\nu=N_a^\nu/N_a^0$, achieved by supplementing our baseline reconfiguration algorithm with the rerouting subroutine to obtain the assignment-rerouting algorithm. If we consider the baseline reconfiguration algorithm, the mean fraction of displaced atoms increases with configuration size~(Fig.~\ref{fig:lossless_rerouting_ordering}a) from $\bar{f}_\nu^{0}=0.65(18)$ for preparing a configuration of $N_a^T=4\times4$ atoms to $\bar{f}_\nu^{0}=0.94(2)$ for preparing a configuration of $N_a^T=32\times32$ atoms.  Nearly all atoms are thus displaced for large configuration sizes, in contrast with the 3-approx algorithm that displaces only slightly more than half of the atoms ($0.55(2)$ for a configuration of
$N_a^T=32\times32$ atoms). 
Compared to the baseline reconfiguration algorithm, the assignment-rerouting algorithm slightly reduces the fraction of displaced atoms~(Fig.~\ref{fig:lossless_rerouting_ordering}b), achieving a mean relative fraction of displaced atoms of $\langle f_\nu^{ar}/f_\nu^{0}\rangle=0.94(^{+6}_{-8})$ for preparing a configuration of $N_a^T=4\times4$ atoms and $\langle f_\nu^{ar}/f_\nu^{0}\rangle=0.983(7)$ for preparing a configuration of $N_a^T=32\times32$ atoms. The mean relative gain in performance is thus larger for smaller configuration sizes, even though a gain in performance is not always possible for small configuration sizes for which the baseline reconfiguration algorithm already minimizes the total number of displaced atoms ($f_{\nu}^{ar}/f_{\nu}^{0}=1$). The gap in the mean fraction of displaced atoms is narrower for larger configuration sizes because we expect longer paths in the path system, which makes the rerouting subroutine less effective, as it becomes less likely that rerouting a path while keeping the rest of the path system unchanged will isolate atoms. The key takeaway of Fig.~\ref{fig:lossless_rerouting_ordering} is that assignment-rerouting is absolutely better than baseline in terms of the fraction of displaced atoms, with the largest relative improvement obtained for small trap arrays; however, its performance does not match the performance of the 3-approx. Any heuristic algorithm exploiting the assignment subroutine would thus be better off using the rerouting subroutine, if not for the increase in computational runtime.

%\subsection{Rerouting-ordering}
%%% \begin fig:lossless_aro
\begin{figure}[t!]
\includegraphics[]{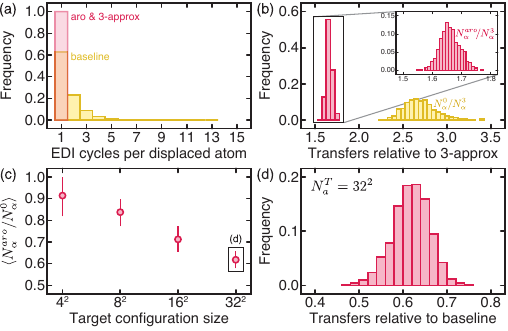}
\caption{
\label{fig:lossless_aro}
\textbf{Reducing the number of transfer operations using the rerouting and ordering subroutines in the absence of loss for a target configuration of $\bm{32\times32}$ atoms.}
(a)~Distribution of the number of EDI cycles per displaced atom using the baseline (yellow), aro (red), and 3-approx (red) algorithms.
(b)~Distribution of the number of transfer operations computed relative to the 3-approx for the baseline (yellow) and aro (red) algorithms. 
(c)~Mean relative number of transfer operations for the aro algorithm computed relative to the baseline reconfiguration algorithm for various configuration sizes.
(d)~Distribution of the relative number of transfer operations for the aro algorithm computed relative to the baseline reconfiguration algorithm. 
}
\end{figure}
%%% \end fig:lossless_aro

We then compute the reduction in the number of EDI cycles per displaced atom obtained by implementing the ordering subroutine. 
The baseline reconfiguration algorithm executes moves in an arbitrary order, possibly displacing the same atom multiple times, and thus having it undergo multiple EDI cycles, each of which entails unnecessary extraction and implantation operations. 
The ordering subroutine improves on the baseline reconfiguration algorithm by ordering the moves so that each atom undergoes at most one EDI cycle~(Fig.~\ref{fig:lossless_aro}a).
The number of transfer operations per displaced atom is strictly reduced to two, as it is the case for the 3-approx algorithm, given one extraction and one implantation operation per EDI cycle.

We further compute the reduction in the total number of transfer operations by implementing the full aro algorithm.
We express the number of transfer operations relative to the 3-approx algorithm, $N^{aro}_\alpha/N^{3}_\alpha$, which performs at most three times the minimum number of transfer operations. For preparing a configuration of $N_a^T=32\times32$ atoms, the baseline reconfiguration algorithm performs on average $2.7(2)$ times more transfer operations than the 3-approx algorithm, whereas the aro algorithm performs on average $1.66(4)$ times more transfer operations~(Fig.~\ref{fig:lossless_aro}b).
The mean relative fraction of transfer operations performed by the aro algorithm over the baseline reconfiguration algorithm decreases with configuration size~(Fig.~\ref{fig:lossless_aro}c), ranging from $0.91(9)$ for preparing a configuration of $N_a^T=4\times4$ atoms to $0.62(4)$ for preparing a configuration of $N_a^T=32\times32$ atoms~(Fig.~\ref{fig:lossless_aro}d). The downward trend can be explained by the deterioration of the efficacy of the baseline reconfiguration algorithm with an increase in the size of the configurations. In large configurations, it is more likely that the target of a move is a vertex in the path associated with another move. Depending on the ordering of the moves, moves that were initially unobstructed may become obstructed, requiring the obstructing atom to be transferred again (See App.~\ref{sec:obstruction_solver}). The aro algorithm remedies this issue by ensuring that each displaced atom is transferred exactly once.

%%% \begin fig:lossless_absolute_comparison
\begin{figure}[t]
\includegraphics[]{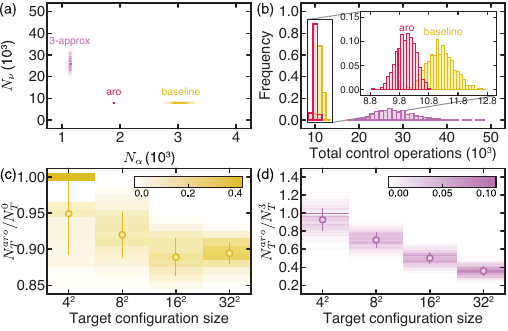}
\caption{
\label{fig:lossless_absolute_comparison}
\textbf{Reducing the total number of control operations using the aro algorithm in the absence of loss for a target configuration of $\bm{32\times32}$ atoms.}
(a)~Distribution of the number of transfer and displacement operations for the baseline (yellow), aro (red), and 3-approx (purple) algorithms.
(b)~Distribution of the number of control operations for the baseline (yellow), aro (red), and 3-approx (purple) algorithms.
(c-d)~Distribution of the relative number of control operations and its mean value for various configuration sizes for the aro algorithm computed relative to the baseline (yellow) and 3-approx algorithms (purple).
}
\end{figure}
%%% \end fig:lossless_absolute_comparison

%Total number of control operations
We finally compute the total number of control operations by summing the number of transfer operations and the number of displacement operations for all atoms~(Fig.~\ref{fig:lossless_absolute_comparison}a).
Because both the baseline and the aro algorithms exactly minimize the number of displacement operations, and the aro algorithm performs fewer transfer operations, the aro algorithm ultimately performs fewer total control operations~(Fig.~\ref{fig:lossless_absolute_comparison}b).
In addition, although the 3-approx algorithm performs fewer transfer operations than both the baseline and aro algorithms, it performs significantly more displacement operations. The 3-approx algorithm thus performs worse than the baseline and aro algorithms in terms of total number of control operations~(Fig.~\ref{fig:lossless_absolute_comparison}b).
The mean relative number of control operations computed for the aro algorithm relative to the baseline (3-approx) algorithm decreases with configuration size, ranging from $0.95(6)$ ($0.93(12)$) for a configuration of $N_a^T=4\times4$ atoms to $0.89(2)$ ($0.36(5)$) for a configuration of $N_a^T=32\times32$
~(Fig.~\ref{fig:lossless_aro}c-d). The uptrend observed in Fig.~\ref{fig:lossless_absolute_comparison}c between $N_a^T = 16 \times 16$ and $N_a^T = 32 \times 32$ is due to displacement operations accounting for most of the control operations when the target configuration size is large. Hence, the gap between the performance of the aro algorithm and that of the baseline algorithm is smaller in terms of the total number of control operations, as both algorithms minimize displacements. The 3-approx algorithm, however, does not minimize displacements, which is reflected in the strictly downward trend in Fig.~\ref{fig:lossless_absolute_comparison}d.

The total number of control operations correlates with the mean success probability in the presence of loss when the duration and the efficiency of control operations are comparable for displacement and transfer operations.
Hence, the relative reduction in the number of control operations achieved by the aro algorithm translates into a relative increase in the mean success probability in the presence of loss, which we quantify in the next section.

\section{Quantifying performance in the presence of loss}~\label{sec:benchmarking_loss}

We numerically evaluate the performance of the aro algorithm in the presence of loss using realistic physical parameters following the approach outlined in our accompanying paper introducing the redistribution-reconfiguration (red-rec) algorithm~(see~\cite{Cimring2023}). 
We conservatively choose the trapping lifetime to be $60~\text{seconds}$ and the success probability of elementary displacement and transfer operations to be $0.985$.
The red-rec algorithm is a heuristic algorithm that seeks to increase operational performance by performing parallel control operations. The key idea of the algorithm is, first, to redistribute atoms from columns containing more atoms than needed (donors) to columns containing fewer atoms than needed (receivers) and, then, to reconfigure each column using an exact 1D reconfiguration algorithm.
In addition to increasing operational performance, it also enables efficient implementation on a low-latency feedback control system with fast computational running time. For the current study, we implement a slightly improved -- more computationally efficient -- version of the red-rec algorithm.

We choose the performance metric to be the mean success probability obtained by averaging the success probability over the distribution of random initial configurations and loss processes. 
In the presence of loss, larger arrays are required to load enough atoms to replace atoms lost during multiple reconfiguration cycles.
As the height of the trap array is increased, there is a sharp transition between near-certain success and near-certain failure~(Fig.~\ref{fig:loss}a-c). 
The relative gain in performance achieved by the aro algorithm over the baseline algorithm is maximized at the inflection point where $\bar{p}=0.5$.

%%% \begin fig:loss
\begin{figure}[t]
\includegraphics[]{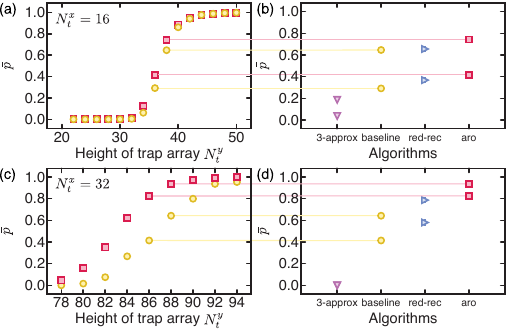}
\caption{
\label{fig:loss}
\textbf{Increasing the mean success probability using the aro algorithm.} Mean success probability, $\bar{p}$, for preparing a configuration of $N_a^T=N_t^x\times N_t^x$ atoms in an array of $N_t=N_t^x\times N_t^y$ traps for 
(a)~$N_t^x=16$, (b)~$N_t^x=16$, $N_t^y=36,38$, (c)~$N_t^x=32$, and (d)~$N_t^x=32$, $N_t^y=86,~88$.
The markers represent the 3-approx (purple inverted triangle), baseline (yellow disk), red-rec (blue triangle), and aro (red square) algorithms.
}
\end{figure}
%%% \end fig:loss

In the presence of loss, the aro algorithm outperforms the 3-approx, baseline, and red-rec algorithms~(Fig.~\ref{fig:loss}b-d). Comparing the baseline and aro algorithms, the mean success probability increases from $\bar{p}_a=0.29(3)$ ($0.65(3)$) to $\bar{p}_{aro}=0.41(4)$ ($0.74(3)$) for a static trap array of $16\times36$ ($16\times38$) traps, and from $\bar{p}_a=0.41(4)$ ($0.64(3)$) to $\bar{p}_{aro}=0.82(3)$ ($0.94(2)$) for a static trap array of $32\times86$ ($32\times88$) traps. The relative gain of performance is $\bar{p}_{aro}/\bar{p}_{a}=1.4(3)$~($1.1(1)$) and $\bar{p}_{aro}/\bar{p}_{a}=2.0(2)$ ($1.5(1)$) for a static trap array of $16\times36$ ($16\times38$) and 
$32\times86$ ($32\times88$) traps, respectively.

The relative improvement in performance is offset by a significant increase in computational running time. The red-rec algorithm thus maintains an operational advantage when real-time computation is needed. Our current implementation of the aro algorithm has, however, not been optimized for real-time operation, and  we foresee opportunities to further improve its running time.

\section{Conclusion}
We have introduced the assignment-rerouting-ordering (aro) algorithm and shown that it outperforms the baseline reconfiguration algorithm, which is a typical assignment-based algorithm, both in the absence and the presence of loss. The aro algorithm exactly minimizes the total number of displacement operations, while reducing the number of displaced atoms and restricting the number of transfer operations to strictly two (one extraction and one implantation operation per EDI cycle). 
%We observe that our improved red-rec heuristic algorithm, which is compatible with real-time implementation, achieves performance comparable to that of the aro algorithm. 

%The aro algorithm thus performs fewer transfer operations than the assignment algorithm, albeit slightly more than what is achieved by the 3-approx algorithm. Still, because the 3-approx algorithm displaces atoms significantly more than the other two algorithms, the aro algorithm outperforms both the assignment and the 3-approx algorithms in terms of the total number of control operations, in the absence of loss, and in terms of mean success probability in the presence of loss. Although operational deployment would require further computational speedup, the aro algorithm provides a reference to  benchmark reconfiguration problems.

Further gains in performance could possibly be achieved by trading off an increase in displacement operations for a decrease in transfer operations. 
As we have seen, to minimize the total number of displacement operations, the baseline reconfiguration algorithm needs to displace nearly all atoms, including those that are already located in the target region of the trap array.
Reducing the number of displaced atoms would reduce the number of transfer operations.
One way to reduce the number of displaced atoms is to impose the constraint that a subset of the atoms located in the target region are fixed in place, i.e., the traps containing them remain idle.
This subset could be selected at random, searched over, or identified based on heuristics, e.g., fixing in place atoms located near the geometric center of the target configuration to minimize their corruption, as they are the most costly to replace.
Another approach is to replace the distance-preserving rerouting subroutine with the distance-increasing rerouting subroutine; as the latter subroutine was designed to trade off an increase in displacement operations for a decrease in transfer operations. 
%For small subsets of atoms, fixing some atoms is unlikely to significantly increase the number of displacement operations. however, for large subsets of atoms, this approach might breakdown when some regions of the target configuration become inaccessible due to domain walls, requiring additional heuristics.

We have alluded to the possibility of using adaptive algorithms to further improve operational performance. By adaptive algorithms, we mean running different reconfiguration algorithms in each reconfiguration cycle, such that the choice of which algorithm to execute is possibly dependent on the measured atom configuration. Deploying adaptive algorithms requires quantifying the performance of the algorithms in terms of how the atoms are distributed in the initial configuration. %This will in turn make it possible to come up with a measure of the performance of the algorithms for a given configuration of atoms.

The current implementation of the aro algorithm has not been optimized for runtime performance, which prevents numerical benchmarking for large configuration sizes and near-term deployment in an operational context. Future work will focus on improving the runtime performance of the aro algorithm, possibly at the cost of integrating additional heuristics, quantifying performance on arbitrary graphs, and demonstrating applicability in an operational setting. 
%As a result of Theorem~\ref{thm:broken_path_system}, we now know of the existence of an $\mathcal{O}(n^4)$ algorithm for the assignment-rerouting-ordering subroutine on arbitrary (positive) edge-weighted graphs, which we briefly described in Sec.~\ref{sec:ordering}.
 
Our results highlight the value to be gained from extending formal results from combinatorial optimization and graph theory in an operational setting.
Additional research opportunities exist in developing exact and approximation algorithms for atom reconfiguration problems that simultaneously optimize multiple objective functions, as well as for problems with  labeled, distinguishable atoms that encode quantum information. 
Besides enabling the preparation of large configurations of atoms, reconfiguration algorithms can also be used to implement quantum information protocols on quantum devices with dynamic connectivity graphs, e.g., by displacing atoms with respect to one another to perform gates among otherwise non-interacting atoms~\cite{Brandhofer2021, Bluvstein2022, Tan2022}.
%These algorithms could be used to design efficient algorithms for implementing quantum circuits, quantum protocols, and quantum error correction codes on quantum devices that admit dynamic connectivity graphs, such as those provided by dynamic configuration of atoms.

Finally, we emphasize that our results are general and thus extend beyond the scope of atom reconfiguration problems. We encourage the interested readers to read our complementary study on the red-rec algorithm~\cite{Cimring2023}. The source code for the benchmarking module and reconfiguration algorithms will be made available upon reasonable request. 

\section{Acknowledgements}
This work was supported by the Canada First Research Excellence Fund (CFREF).
Amer E. Mouawad's work was supported by the Alexander von Humboldt Foundation and partially supported by the PHC Cedre project 2022 ``PLR''. Research by Remy El Sabeh, Stephanie Maaz, and Naomi Nishimura was supported by the Natural Sciences and Engineering Research Council of Canada.

\begin{appendix}

\begin{figure}[t]
\includegraphics[]{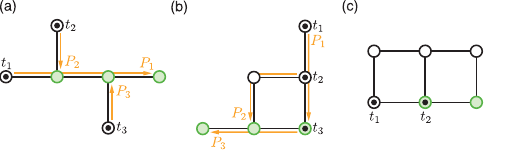}
\caption{
\label{fig:fig1a}
\textbf{Examples of atom reconfiguration problems on graphs.}
(a)~Example problem for which ordering the paths of a path system reduces the number of transfer operations. %The indices of the source vertices and the target vertices indicate the index of the path they are associated with, and since the graph is a tree, the corresponding paths can be inferred from the graph.
Executing the move associated with either $P_2$ or $P_3$ (or both) before the move associated with $P_1$ would force $\tau_2$ or $\tau_3$ (or both) to move twice.
(b)~Example problem for which a token (here, $\tau_3$) can be isolated without increasing the weight of the path system. %The red vertices are target vertices, the black disks are tokens (therefore the vertices with black circles in them are source vertices), and the colored arrows are paths in the path system.
(c)~Example problem for which the displacement distance and the number of displaced tokens cannot be simultaneously minimized. Two tokens need to move to minimize the displacement distance, whereas one token needs to move if minimizing displacements is not imposed.
}
\end{figure}

\section{Assignment subroutine}\label{app:assignment}
The assignment subroutine first solves the \emph{all-pairs shortest path} (APSP) problem in order to compute the shortest pairwise distances between occupied traps in the initial configuration (represented as the set of source vertices $S$) and occupied traps in the target configuration (represented as the set of target vertices $T$). More generally, the algorithm consists of computing shortest paths between every vertex in $S$ and every vertex in $T$. 
On uniformly-weighted grid graphs, the APSP problem can be solved in a time that is quadratic in the number of vertices ($\mathcal{O}(n^2)$), as the shortest distance between any two traps in the grid is equal to the Manhattan distance between them.
In the general edge-weighted case, this problem can be solved in $\mathcal{O}(n^3)$ time using the Floyd-Warshall algorithm~\cite{floyd1962algorithm}.
Alternatively, because it is more easily amenable to parallel implementation, e.g., on a GPU, Dijkstra's algorithm (run from each source vertex) could replace the Floyd-Warshall algorithm to find shortest paths between pairs of vertices in positive edge-weighted graphs.

The assignment subroutine then solves the assignment problem to find one of possibly many sets of pairs of source and target traps that exactly minimize the total distance traveled by all atoms, and thus exactly minimize the total number of displacement operations. The problem of computing a distance-minimizing path system can be reduced to solving an assignment problem using the assignment subroutine. A solution to the \emph{assignment problem} consists of finding a bijection $f: A \rightarrow B$ that minimizes the total cost $\sum_{a \in A} C(a, f(a))$, 
where $A$ and $B$ are two sets of equal cardinality and $C: A \times B \rightarrow \mathbb{R}$ is a positive cost function. In the absence of a surplus of atoms, i.e., when $|S| = |T|$, we set $A = S$, $B = T$, and for any pair $a \in A$, $b \in B$, we choose the cost function $C(a, b)$ to be the distance between vertex $a$ and vertex $b$. Clearly, a solution to the assignment problem is a matching between source traps and target traps that minimizes the total displacement distance. This assignment problem has a polynomial-time solution; the first documented solution, which was attributed to Kuhn as the Hungarian algorithm~\cite{kuhn1955hungarian}, has an asymptotic running time of $\mathcal{O}(|A|^4)$ that was later improved to $\mathcal{O}(|A|^3)$~\cite{edmonds1972theoretical}.

In our current implementation of the assignment subroutine, we solve the APSP problem using the Floyd-Warshall algorithm~\cite{floyd1962algorithm}, which we modify to store one of the shortest paths between each pair of source and target vertices. We then use the Hungarian algorithm to solve the assignment problem and find one of many possible pairings of source and target traps that minimize the total displacement distance; each pair is associated with the shortest path computed and stored earlier.  

We now show how the problem of computing a $T$-valid distance-minimizing path system (in  positive edge-weighted graph) can be reduced to the assignment problem, even with the existence of surplus atoms. %In the absence of token surplus, i.e., when $|S| = |T|$, we set $A = S$, $B = T$, and for any pair $a \in A$, $b \in B$, we set $C(a, b)$ to be the distance between vertex $a$ and vertex $b$. Clearly, the matching obtained from running the Hungarian algorithm minimizes total distance. 
%Because we would like to solve reconfiguration problems with surplus atoms, we make a small modification in the reduction. 
Given a positive edge-weighted graph $G$, and two sets $S, T \subseteq V(G)$ such that $|S| \geq |T|$, we define a set $U$ ($|U| = |S| - |T|$), and we set $A = S$ and $B = T \cup U$. As for the cost function, we define it for any tuple in $A \times B$ as 
\[ C(a, b) = 
    \begin{cases} 
      d_G(a, b) & b \in T\\
      W & \text{otherwise,} \\
   \end{cases}
\]
\noindent
where $d_G(a,b)$ is the weight of a shortest path between $u$ and $v$ in $G$ and $W$ is a large number (we set $W$ large enough to ensure that it is larger than $|S|$ times the weight of the heaviest shortest path). Running the Hungarian algorithm on the constructed instance will yield an assignment such that each vertex in $B \setminus U$ is assigned to a source vertex in $A$; this subset of the matching can be used to construct a valid distance-minimizing path system $\mathcal{P}$ (where for each pair we pick one of the many possible shortest paths)~\cite{Calinescu2007}. We provide a proof of correctness for completeness. 

\begin{lemma}
Given an $n$-vertex (positive) edge-weighted graph $G$ and two sets $S,T \subseteq V(G)$ such that $|S| \geq |T|$, we can compute in time $\mathcal{O}(n^3)$ a $T$-valid distance-minimizing path system $\mathcal{P}$. 
\end{lemma}

\begin{proof}
Note that $C(a, b)$ is equal to the weight
of a shortest path in $G$ connecting $a$ and $b$ whenever $a \in S$ and $b \in T$ and $W$ otherwise. After we obtain the matching, we can move the tokens matched to vertices in $B \setminus U$ as follows. Assume we want to move token $\tau_i$; if the path $\tau_i$ would take to reach its target has another token $\tau_j$ on it, we switch the targets of the two tokens and we move $\tau_j$ instead, in a manner similar to the obstruction solver subroutine. One can check that the weight of the edges traversed does not exceed the weight of the path system. On the other hand, the optimum solution (in fact, any solution) must move tokens to targets and cannot do better than the total weight of the shortest paths in a minimum-weight assignment.
The running time follows from the fact that we run the Floyd-Warshall algorithm followed by the Hungarian algorithm, each of which has a running time of $\mathcal{O}(n^3)$.
\end{proof}

\section{Isolation subroutine}\label{app:isolation}
A common issue associated with the assignment subroutine is that it might label a vertex as both a source and a target vertex, either labeling a target vertex as its own source or labeling a vertex as the target of one path and the source of another. 
Such double labeling might result in unnecessary transfer operations, possibly displacing an atom that could otherwise have remained in place.
To minimize the number of unnecessary transfer operations caused by double labeling, we implement an isolation subroutine, which we run before computing the moves associated with the path system, to remove doubly-labeled vertices whenever possible. The removal of a doubly-labeled vertex is possible whenever the recomputed path system (obtained after excluding the vertex and its incident edges from the graph and updating $S$ and $T$ accordingly) remains valid and has a total weight that is less than or equal to the total weight of the original path system. The isolation subroutine guarantees that every token in the resulting path system has to be displaced at least once. 
It also guarantees that the baseline reconfiguration algorithm does not serendipitously displace fewer tokens than the aro algorithm; however, because implementing the subroutine is computationally costly and these serendipitous instances are rare, this subroutine can be safely ignored in an operational setting.

As just explained, the isolation subroutine is a heuristic that guarantees that, for any distance-minimizing path system, every token in the path system (after deleting some vertices and tokens from the graph) has to move at least once, regardless of the path system or of the order in which we execute the moves. 

We let $I_{\mathcal{P}}$ denote the set of vertices which contain tokens that are isolated in $\mathcal{P}$, and we let $F_{\mathcal{P}}$ denote the set of vertices in $S$ that do not appear in $\mathcal{P}$ (their tokens are said to be \emph{fixed in place by $\mathcal{P}$}). We show, in particular, that we can transform $\mathcal{P}$ into another path system $\mathcal{P'}$ such that $w(\mathcal{P'}) = w(\mathcal{P})$, $I_{\mathcal{P'}} \supseteq I_{\mathcal{P}}$, $F_{\mathcal{P'}} \supseteq F_{\mathcal{P}}$, and there exists no other path system $\mathcal{P''}$ such that $w(\mathcal{P''}) = w(\mathcal{P'})$, and either $I_{\mathcal{P''}} \supset I_{\mathcal{P'}}$ or $F_{\mathcal{P''}} \supset F_{\mathcal{P'}}$ (see Figure~\ref{fig:fig1a}b for an example of an instance where our rerouting heuristics would fail to isolate an extra token while the isolation subroutine succeeds).

Assume we are given a distance-minimizing path system $\mathcal{P}$ and let $v \in S \subseteq V(G)$ be a vertex containing a token $\tau$ which is not isolated or fixed in place in $\mathcal{P}$ (not in $I_{\mathcal{P}} \cup F_{\mathcal{P}}$). Let $G' = G - (I_\mathcal{P} \cup F_\mathcal{P} \cup \{v\})$ denote the graph obtained from $G$ after deleting the set of vertices $I_\mathcal{P} \cup F_\mathcal{P} \cup \{v\} \subseteq V(G)$ and the edges incident on all vertices in $I_\mathcal{P} \cup F_\mathcal{P} \cup \{v\}$.  We compute an assignment (App.~\ref{app:assignment}) in the graph $G' = G - (I_\mathcal{P} \cup F_\mathcal{P} \cup \{v\})$ with $S' = S \setminus (I_\mathcal{P} \cup F_\mathcal{P} \cup \{v\})$ and $T' = T \setminus (I_\mathcal{P} \cup F_\mathcal{P} \cup \{v\})$.
Let $\mathcal{P'}$ denote the path system associated with this newly computed assignment. If $\mathcal{P'}$ is $T'$-valid and $w(\mathcal{P'}) = w(\mathcal{P})$, then we use $\mathcal{P'}$ instead of $\mathcal{P}$ and consider token $\tau$ (on vertex $v$) as either an extra isolated token or an extra fixed in place token, depending on whether $v \in S \cap T$ or $v \in S \setminus T$. We then add $v$ to either $I_{\mathcal{P'}}$ or $F_{\mathcal{P'}}$, depending on whether the token on it was isolated or fixed in place. This process is repeated as long as we can find new tokens to isolate or fix in place. We show that, in a distance-minimizing path system where no tokens can be isolated or fixed in place, all other tokens will have to move at least once. We call such a path system an \emph{all-moving path system}.

\begin{lemma}\label{lem-greedy-isolate}
Given an $n$-vertex (positive) edge-weighted graph $G$, two sets $S,T \subseteq V(G)$ such that $|S| \geq |T|$, and a $T$-valid distance-minimizing path system $\mathcal{P}$, we can compute, in time $\mathcal{O}(n^5)$, a valid all-moving path system $\mathcal{P'}$ such that $w(\mathcal{P'}) = w(\mathcal{P})$, $I_{\mathcal{P'}} \supseteq I_{\mathcal{P}}$, and $F_{\mathcal{P'}} \supseteq F_{\mathcal{P}}$. 
\end{lemma}

\begin{proof}
Let $\mathcal{P'}$ denote the path system obtained after exhaustively applying the described procedure. Clearly, $\mathcal{P'}$ is valid and $w(\mathcal{P'}) = w(\mathcal{P})$ (by construction). It remains to show that $\mathcal{P'}$ is an all-moving path system. In other words, we show that there exists no other valid path system $\mathcal{P''}$ such that $w(\mathcal{P''}) = w(\mathcal{P'})$ and $\mathcal{P''}$ can isolate or fix a proper superset of $I_{\mathcal{P'}}$ or $
F_{\mathcal{P'}}$, i.e., $I_{\mathcal{P''}} \supset  I_{\mathcal{P'}}$ or $F_{\mathcal{P''}} \supset F_{\mathcal{P'}}$.

Assume that $\mathcal{P''}$ exists and let $v^\star \in I_{\mathcal{P''}} \setminus I_{\mathcal{P'}}$ or $v^\star \in F_{\mathcal{P''}} \setminus F_{\mathcal{P'}}$. In either case, we know that $v^\star \in S$, which gives us the required contradiction as our procedure would have deleted $v^\star$ (and either isolated or fixed the token on it). 

For the running time, note that we can iterate over vertices in $S$ in $\mathcal{O}(n)$ time. Once a vertex is deleted, we run the APSP algorithm followed by the Hungarian algorithm, and this requires $\mathcal{O}(n^3)$ time. We can delete at most $n$ vertices, and whenever a vertex is successfully deleted (which corresponds to fixing in place or isolating a token), we repeat the procedure. Therefore, the total running time of the isolation procedure is $\mathcal{O}(n^5)$. 
\end{proof}

\section{Distance-preserving rerouting subroutine}\label{app:rerouting}

Following Sec.~\ref{sec:rerouting} where we present an overview of the distance-preserving rerouting subroutine, we show that the problem of rerouting a path system to increase the number of isolated tokens without increasing the path system's total weight has a substructure that we can utilize to design a polynomial-time dynamic programming solution.  We present this subroutine in detail for uniformly-weighted grid graphs and we sketch how it could be generalized to work on general positive edge-weighted graphs.

We assume that we are working with a path $P$ with a source vertex $(c_1, r_1)$ and a target vertex $(c_2, r_2)$ such that $r_1 < r_2$ and $c_1 < c_2$. The other three cases entail flipping one or both of these inequalities and can be solved in a similar way. For the path in question, we introduce a $(W + 1) \times (H + 1)$ matrix, $dp$, such that $W = |c_1 - c_2|$, $H = |r_1 - r_2|$, and $dp[i][j]$ is the smallest number of isolated tokens on any shortest path between $(c_1, r_1)$ and $(c_1 + i, r_1 + j)$ in the path system that excludes the current path $P$, i.e., in $\mathcal{P} \setminus P$. The $dp[i][j]$ values are computed as follows:

%https://tex.stackexchange.com/questions/588416/writing-left-side-of-the-cases-brackets 
\begin{equation*}
\begin{split}
       & dp[i][j] = \\
       & 
\begin{matrix}
  \textsf{I}    \\[0.5ex]
  \textsf{II}   \\[0.5ex]
  \textsf{III}  \\[0.5ex]
  \textsf{IV}  \\[0.5ex]
  \\
\end{matrix}\quad
\begin{cases}
     0 & (i = 0, j = 0) \\
  \text{isolated}[i][j] + dp[i-1][j] & (i \neq 0, j = 0) \\
  \text{isolated}[i][j] + dp[i][j-1] & (i = 0, j \neq 0) \\
  \text{isolated}[i][j] + \\
  ~~\min(dp[i-1][j], dp[i][j-1]) & \text{otherwise}
\end{cases}
\end{split}
\end{equation*}

Here $\text{isolated}[i][j]$ is equal to $1$ whenever there is an isolated token on vertex $(c_1 + i, r_1 + j)$ in $\mathcal{P} \setminus P$ and $0$ otherwise. 
The value of interest is $dp[W][H]$, which can be computed in $\mathcal{O}(WH)$ time. The proof of correctness of this algorithm follows from the next lemma. We note that the generalization to positive edge-weighted graphs entails modifying the Floyd-Warshall algorithm to break ties between shortest paths based on the number of isolated tokens on them in the path system that excludes them. It follows that, in the general case, the smallest number of isolated tokens on any shortest path between two vertices can be computed in $\mathcal{O}(n^3)$ time.

The statement and the proof of the following lemma assume that $r_1 < r_2$ and $c_1 < c_2$. The proof can be altered to work for the other three possible cases.

\begin{lemma}\label{lem-dp}
Given a valid path system $\mathcal{P}$ in a uniformly-weighted grid graph $G$ and a path $P \in \mathcal{P}$ with source vertex $(c_1, r_1)$ and target vertex $(c_2, r_2)$ such that $r_1 < r_2$ and $c_1 < c_2$, $dp[i][j]$ is the smallest number of isolated tokens in the path system that excludes the path $P$ (i.e., $\mathcal{P} \setminus P$) on any shortest path between the source vertex and vertex $(c_1 + i, r_1 + j)$.
\end{lemma}

\begin{proof}
We use induction on the two indices $i$ and $j$. For the path going from the source vertex $(c_1, r_1)$ to itself, there are no tokens in $\mathcal{P} \setminus P$ since the current path is excluded from the path system, so $dp[0][0] = 0$ (case I) is correct. Similarly, for the vertices on the same column and the same row as $(c_1, r_1)$, there is a single distance-minimizing path to each of them, so $\text{isolated}[i][j] + dp[i-1][j]$ and $\text{isolated}[i][j] + dp[i][j - 1]$, respectively, compute the numbers of isolated tokens on the paths to those vertices correctly (cases II and III).
Now, assume that $dp[i][j]$ takes on the correct value for $0 \leq i \leq k$, $0 \leq j \leq k'$ (excluding the pair ($k, k'$)). 
%That is, $dp[i][j]$ is the smallest number of isolated tokens on any path between the source vertex and vertex $(c_1 + i, r_1 + j)$ in the specified intervals (excluding the pair ($k, k'$)). 
We prove that $dp[k][k']$ takes on the correct value. 

In any path from $(c_1, r_1)$ to $(c_1 + k, r_1 + k')$, vertex $(c_1 + k, r_1 + k')$ can be reached on a distance-minimizing path either from vertex $(c_1 + k - 1, r_1 + k')$ or from vertex $(c_1 + k, r_1 + k' - 1)$. By the inductive hypothesis, we already know the smallest number of isolated tokens from $(c_1, r_1)$ to either one of those two vertices in $\mathcal{P} \setminus P$; the smallest number of isolated tokens from $(c_1, r_1)$ to $(c_1 + k, r_1 + k')$ will therefore be the minimum of those two values, to which we add 1 in case there is an isolated token on vertex $(c_1 + k, r_1 + k')$ in $\mathcal{P} \setminus P$ (case IV).
\end{proof}

Once we obtain $dp[W][H]$, if its value is smaller than the number of isolated tokens along the current path $P$ in the path system that excludes $P$, then $P$ has to be rerouted (since we can isolate more tokens by rerouting $P$); otherwise, $P$ is unchanged. Since path reconstruction is required, we need to store the decisions that were made by the dynamic programming procedure, and the easiest way to do so is by introducing a $(W + 1) \times (H + 1)$ matrix, $prev$, such that $prev[i][j]$ indicates whether $dp[i][j]$'s value was obtained by reaching vertex $(c_1 + i, r_1 + j)$ from the bottom ($prev[i][j] = 1$) or from the left ($prev[i][j] = 0$). Using the $prev$ matrix, we can therefore reconstruct the rerouted path and use it to replace the initial path if needed.

\begin{lemma}
\label{distance_preserving_running_time}
Given a valid path system $\mathcal{P}$ in an $n$-vertex uniformly-weighted grid graph $G$ (resp. positive edge-weighted graph $G$), we can, in time $\mathcal{O}(n^3)$ (resp. in time $\mathcal{O}(n^5)$), exhaustively run the distance-preserving rerouting heuristic.
\end{lemma}

\begin{proof}
The procedure loops over all paths in the path system and attempts to reroute each path. If at least one path is rerouted, once all paths have been considered, the process is repeated. The process keeps getting repeated until the algorithm goes through all paths without changing any path. The algorithm terminates in $\mathcal{O}(n^3)$ time on uniformly-weighted grid graphs (resp. $\mathcal{O}(n^5)$ time on positive edge-weighted graphs),
assuming at most $n$ paths in the path system, given that we can isolate at most $n$ tokens, that computing a replacement path requires $\mathcal{O}(WH) = \mathcal{O}(n)$ time (resp. $\mathcal{O}(n^3)$ time), and that every time a token is isolated, the procedure repeats from scratch.
\end{proof}

\section{Distance-increasing rerouting subroutine}\label{app:distance-increasing-rerouting}
Following Sec.~\ref{sec:rerouting} where the distance-increasing rerouting subroutine is briefly mentioned, we now provide a detailed presentation. This subroutine is a trade-off heuristic that seeks to trade an increase in displacement distance for a reduction in the number of displaced tokens. We present this subroutine only for uniformly-weighted grid graphs and we note that the distance-increasing rerouting subroutine could potentially be generalized to work on general positive edge-weighted graphs. We say that a path in a grid is \emph{rectilinear} if it is horizontal or vertical (assuming an embedding of the grid in the plane). 
Recall that we say that a token $\tau$ is isolated in a path system $\mathcal{P}$ whenever there exists a single-vertex path $P = \{v\} \in \mathcal{P}$ such that the token $\tau$ is on the vertex $v$ and no other path in $\mathcal{P}$ contains vertex $v$. 

The purpose of the distance-increasing rerouting subroutine is to introduce a mechanism that allows us to increase token isolation, even if that comes at the cost of increasing displacement distance. We also want to make it possible to control how much leeway is given to this subroutine when it comes to deviating from the minimization of overall displacement distance. 
To do so, we introduce the concept of a \emph{margin}, which we denote by $\mu$. The margin limits the paths considered. For a margin $\mu$, a source vertex $(c_1, r_1)$, and a target vertex $(c_2, r_2)$ (assuming, without loss of generality, $r_1 < r_2$ and $c_1 < c_2$) defined in a path system in a uniformly-weighted grid graph $G$, the rerouted path can now include any of the vertices that are within the subgrid bounded by the vertices $(max(c_1 - \mu, 0), max(r_1 - \mu, 0))$ (bottom left corner) and $(min(c_2 + \mu, W_G - 1), min(r_2 + \mu, H_G - 1))$ (top right corner), which we call the \emph{extended subgrid}. 

As was the case for the analysis of the distance-preserving rerouting subroutine, for the rest of this section, we assume that we are working with a path system $\mathcal{P}$ and a path $P$ with a source vertex $(c_1, r_1)$ and a target vertex $(c_2, r_2)$ such that $r_1 < r_2$ and $c_1 < c_2$. The other three cases can be handled in a similar way. Evidently, as was the case for the distance-preserving rerouting subroutine, there is no point in attempting to enumerate all the paths here either, as their number is exponential in $H + W + \mu$. In fact, even for $\mu = 0$, the possible reroutings are a superset of the possible reroutings in the distance-preserving rerouting, as we removed the restriction on maintaining a shortest path within the subgrid.

Out of the reroutings of $P$ that minimize the number of isolated tokens in $\mathcal{P} \setminus P$, we select the ones that have the shortest path length, and out of those, we select the ones that have the smallest number of changes in direction (horizontal vs. vertical). We arbitrarily select any one of the remaining paths.

Again, we exploit dynamic programming to solve the problem. Just like in App.~\ref{app:rerouting}, we make use of the matrix $\text{isolated}$; however, in this case, we have to consider all vertices that are part of the extended subgrid, so the matrix $\text{isolated}$ is of size $(W + 1 + 2\mu) \times (H + 1 + 2\mu)$, and $\text{isolated}[i][j]$ is equal to $1$ whenever $(c_1 + i - \mu, r_1 + j - \mu)$ is in the grid and contains an isolated token, and $0$ otherwise. We introduce a $((W + 1 + 2\mu) \times (H + 1 + 2\mu)) \times (W + 1 + 2\mu) \times (H + 1 + 2\mu)$ matrix, $dp$, such that $dp[i][j][k]$ is the smallest number of isolated tokens on any path of length $i$ between $(c_1, r_1)$ and $(c_1 + j - \mu, r_1 + k - \mu)$ in $\mathcal{P} \setminus P$. The $dp[i][j][k]$ values are computed as follows:

\begin{equation*}
    \begin{split}
       & dp[i][j][k] = \\
       & 
\begin{matrix}
  \textsf{I}    \\[0.5ex]
  \textsf{II}   \\[0.5ex]
  \textsf{III}  \\[0.5ex]
  \textsf{IV}  \\[0.5ex]
  \\
  \textsf{V}  \\[0.5ex]
  \\
  \\
  \\
\end{matrix}\quad
\begin{cases}
    0 ~~ (i = 0, j = \mu, k = \mu) \\
    +\infty ~~ ((c_1 + j - \mu, r_1 + k - \mu) \text{ not in grid graph}) \\
      +\infty ~~ (i = 0, j \neq \mu \text{ or } k \neq m) \\
      +\infty ~~ (dp[i-1][j - 1][k] = dp[i - 1][j + 1][k] \\
      ~~= dp[i - 1][j][k - 1] = dp[i - 1][j][k + 1] = +\infty) \\
      \text{isolated}[j][k] +  
      \\~~ \min(dp[i-1][j-1][k], dp[i-1][j+1][k],\\
      ~~dp[i-1][j][k-1], \\ ~~dp[i-1][j][k+1])~\text{otherwise} \\
   \end{cases}
    \end{split}
\end{equation*}

By convention, we set $dp[i][j][k]$ to $+\infty$ when (1) the shortest distance between $(c_1, r_1)$ and $(c_1 + j - \mu, r_1 + k - \mu)$ is greater than $i$ (cases III and IV, which is equivalent to saying that the smallest number of isolated tokens of any path of length $i$ between those two vertices in $\mathcal{P} \setminus P$ is infinite, and this makes case V work), (2) when $(c_1 + j - \mu, r_1 + k - \mu)$ are coordinates that do not correspond to the coordinates of a vertex within the grid graph (case II), or (3) when at least one of $j$ or $k$ is out of bounds, i.e., when $(c_1 + j - \mu, r_1 + k - \mu)$ has fewer than 4 neighboring vertices in the grid graph (this makes it easier to deal with cases IV and V).

The value of interest is $\min(dp[W + H][\mu + W][\mu + H], dp[W + H + 1][\mu + W][\mu + H], \ldots, dp[(W + 1 + 2\mu) \times (H + 1 + 2\mu) - 1][\mu + W][\mu + H])$. That is, we consider the maximum number of tokens we managed to isolate for every path length greater than or equal to $W + H$ (which is the length of the shortest path between ($c_1, r_1$) and ($c_2, r_2$)) and we pick the maximum across all path lengths. The proof of correctness of this procedure follows from Lemma~\ref{bounded_correctness}. 

The statement and the proof assume that $r_1 < r_2$, $c_1 < c_2$. The proof can be altered to work for the other three possible cases.

\begin{lemma}
\label{bounded_correctness}
Given a valid path system $\mathcal{P}$ in a uniformly-weighted grid graph $G$ and a path $P \in \mathcal{P}$ with source vertex $(c_1, r_1)$ and target vertex $(c_2, r_2)$ such that $r_1 < r_2$ and $c_1 < c_2$, $\min\limits_{i \in \mathbb{N}^{+}([0,(W+1+2\mu)\times(H+1+2\mu)-1])}dp[i][j][k]$ is the smallest number of isolated tokens in $\mathcal{P} \setminus P$ on any path between the source vertex and vertex $(c_1 + j - \mu, r_1 + k - \mu)$.
\end{lemma}

\begin{proof}
We start by proving the correctness of the dynamic programming approach, as the correctness of the lemma statement follows directly from that. That is, we start by proving that $dp[i][j][k]$, for pairs $(j, k)$ that correspond to vertices in the grid, is the smallest number of isolated tokens on any path of length $i$ between $(c_1, r_1)$ and $(c_1 + j - \mu, r_1 + k - \mu)$ in the path system that excludes the path $P$, and is equal to $+\infty$ otherwise.

%When computing $dp[i][j][k]$, we use the value $-1$ to signify that $(c_1 + j - m, r_1 + k - m)$ cannot be reached from $(c_1, r_1)$ within $i$ displacements or is out of bounds. 
We will use induction on the first dimension only (i.e., the path length dimension/number of edges in the path).

The only path of length 0 starting from $(c_1, r_1)$ reaches $(c_1, r_1)$. Since the current path is excluded from the path system, there are no tokens from $(c_1, r_1)$ to itself, and therefore $dp[0][\mu][\mu] = 0$ (case I), as needed. No other vertex is reachable for $i = 0$, so $dp[0][i][j]$ ($0 \leq i \leq W + 2\mu$, $0 \leq j \leq H + 2\mu$) should equal $+\infty$, and case III ensures that $dp[0][i][j]$ takes on the correct value in the specified range.

Now, assume that $dp[l][j][k]$ takes on the correct values ($0 \leq j \leq W + 2\mu$, $0 \leq k \leq H + 2\mu$). We would like to prove that $dp[l + 1][j][k]$ takes on the correct values ($0 \leq j \leq W + 2\mu$, $0 \leq k \leq H + 2\mu$).

There are two cases to consider. For a fixed value of the pair $(j, k)$ corresponding to a vertex $(c_1 + j - \mu, r_1 + k)$ in the grid, if none of $(c_1 + j - \mu - 1, r_1 + k - \mu)$, $(c_1 + j - \mu + 1, r_1 + k - \mu)$, $(c_1 + j - \mu, r_1 + k - \mu - 1)$ and $(c_1 + j - \mu, r_1 + k - \mu + 1)$ is reachable within $l$ displacements, then $(c_1 + j - \mu, r_1 + k - \mu)$ should not be reachable within $l + 1$ displacements. By the induction hypothesis, we would have $dp[l][j-1][k] = dp[l][j+1][k] = dp[l][j][k-1] = dp[l][j][k+1] = +\infty$, which sets the value of $dp[l+1][j][k]$ to $+\infty$ as well, as needed (case IV).

The second case is when at least one of the neighbors of $(c_1 + j - \mu, r_1 + k - \mu)$ is reachable in $l$ steps. By the inductive hypothesis, we have the smallest number of isolated tokens in $\mathcal{P} \setminus P$ from $(c_1, r_1)$ to the neighbors of $(c_1 + j - \mu, r_1 + k - \mu)$ reached in $l$ steps. The vertex $(c_1 + j - \mu, r_1 + k - \mu)$ can only be reached in $l + 1$ steps from any of its neighbors that were reached in $l$ steps, so the smallest number of isolated tokens on any path of length $l + 1$ from $(c_1, r_1)$ to $(c_1 + j - \mu, r_1 + k - \mu)$ in $\mathcal{P} \setminus P$ is equal to the minimum of the values obtained for the neighbors reached in $l$ steps, to which we add 1 in the case in which there is an isolated token on vertex $(c_1 + j - \mu, r_1 + k - \mu)$, which is what the algorithm does (case V). 

We have proven that $dp[i][j][k]$ takes on the correct value, that is, $dp[i][j][k]$ is the smallest number of isolated tokens on any path of length $i$ between $(c_1, r_1)$ and $(c_1 + j - \mu, r_1 + k - \mu)$ in $\mathcal{P} \setminus P$. It follows that the smallest number of isolated tokens on any path between $(c_1, r_1)$ and $(c_1 + j - \mu, r_1 + k - \mu)$ in $\mathcal{P} \setminus P$ is $\min\limits_{i \in \mathbb{N}^{+}([0,(W+1+2\mu)\times(H+1+2\mu)-1])}dp[i][j][k]$.
\end{proof}

Just like the distance-preserving case, the distance-increasing case requires keeping track of paths, their lengths, as well as their number of changes in direction; $dp$ can be easily augmented to accommodate for that, or an auxiliary matrix (similar to $prev$ in the previous subsection) can be used to keep track of this information. The procedure is run exhaustively, that is, we loop over all paths and attempt to reroute them, and if at least one path is rerouted, once all paths have been considered, the process is repeated. It remains to prove that the procedure terminates and runs in polynomial time. 

\begin{lemma}
\label{bounded_termination_increasing}
Given a valid path system $\mathcal{P}$ in an $n$-vertex uniformly-weighted grid graph $G$ and a margin $\mu \leq n$, we can, in time $\mathcal{O}(n^7)$, exhaustively run the distance-increasing rerouting heuristic.
\end{lemma}

\begin{proof}
Computing all the entries in the augmented $dp$ table for a single path can be achieved in time $\mathcal{O}(n^3)$. Now, every time a path is rerouted we have one of the following three consequences:
\begin{enumerate}
    \item The isolation of one or more tokens and an indeterminate effect on the overall weight of the path system and the overall number of direction changes in the path system
    \item The decrease of the overall weight of the path system and an indeterminate effect on the overall number of direction changes in the path system (no increase in the number of isolated tokens)
    \item The decrease of the overall number of direction changes in the path system (no increase in the overall weight of the path system or decrease in the number of isolated tokens)
\end{enumerate}

The number of tokens that can be isolated is linear in the number of vertices in the grid, whereas the overall weight of the path system as well as the overall number of direction changes in the path system are quadratic in the number of vertices in the grid.

On the one hand, the weight of a path can be decreased at most $\mathcal{O}(n)$ times, ignoring the effect that token isolation may have on the weight of the path. On the other hand, the number of direction changes in a path can be decreased at most $\mathcal{O}(n)$ times, ignoring the effect that token isolation or a decrease in path weight may have on the number of direction changes. Now, since decreasing the weight of a path may affect the number of direction changes within it, in the worst case, every two reroutings that decrease the weight of a path $P$ (of which we have $\mathcal{O}(n)$) may be separated by $\mathcal{O}(n)$ reroutings of path $P$, each of which decreases its number of direction changes. Therefore, ignoring token isolation, each path can be rerouted at most $\mathcal{O}(n^2)$ times, for a total of $\mathcal{O}(n^3)$ path reroutings. 

We now incorporate token isolation into our analysis, and we consider how it interacts with the other two consequences of distance-increasing rerouting. Since each token isolation may lead to updating the weight of a single path to $\mathcal{O}(n)$, accounting for the $\mathcal{O}(n)$ token isolations and relying on the reasoning from the previous paragraph means that we may have $\mathcal{O}(n^3)$ extra path reroutings in total, in addition to the $\mathcal{O}(n^3)$ path reroutings we have already accounted for earlier. This implies that the total number of possible reroutings is bounded by $\mathcal{O}(n^3)$. 

It remains to show how fast we can accomplish each rerouting. Recall that computing all the entries in the augmented $dp$ table for a single path can be achieved in time $\mathcal{O}(n^3)$.
Moreover, the algorithm reiterates through the paths of the path system from scratch every time a rerouting takes place, meaning that a single rerouting is completed in time $\mathcal{O}(n^4)$. 
Given that the total number of possible reroutings is $\mathcal{O}(n^3)$, we can exhaustively run the distance-increasing rerouting heuristic in time $\mathcal{O}(n^7)$. 
\end{proof}

\section{Ordering subroutine}\label{app:ordering}

The ordering subroutine (briefly explained in Sec.~\ref{sec:ordering}) seeks to order the path system and convert it into a (valid) ordered path system. It is the central module of the aro algorithm and consists of finding the order in which to execute the moves so that each token is displaced at most once~(Fig.~\ref{fig:fig1a}a).

\subsection{Step 1 -- Merge path system}\label{app:mps}
The first step of the aro subroutine is to compute a  \emph{merged path system} (MPS). A merged path system is a path system such that no pair of paths within the system intersects more than once, where an intersection between two paths is a nonempty, maximal sequence of vertices that appear in each path's vertex sequence representation contiguously either in the same order or in reverse order.

\begin{lemma}[Merged path system]\label{lemma:mps}
Given an $n$-vertex (positive) edge-weighted graph $G$ with $|E(G)| = m$ and $\sum_{e \in E(G)}{w(e)} = \mathcal{O}(n^{c})$ for some positive integer $c$, two sets $S,T \subseteq V(G)$, and a $T$-valid path system $\mathcal{P}$, we can compute, in time $\mathcal{O}(n^{c+4}m^2)$, a valid merged path system $\mathcal{P'}$ such that $w(\mathcal{P'}) \leq w(\mathcal{P})$. Moreover, the number of distinct edges used in $\mathcal{P'}$ is at most the number of distinct edges used in $\mathcal{P}$. 
\end{lemma}

\begin{proof}
% For a path $P_i$ that intersects path $P_j$, let $v_{p}$ and $v_{q}$ be the first and the last vertex in $P_i$ that are also in $P_j$, respectively. Consider the two subpath $P_i'$ and $P_j'$ that start at $v_p$ and end at $v_q$. An edge in a pair of intersecting paths $P_i$ and $P_j$ is said to be \emph{unique} if it belongs to either $P_{i}'$ or $P_{j}'$ but not both.
For any pair of paths $P_i$, $P_j$ that intersect, let $v_{i,f}$ and $v_{i,l}$ be the first and last vertices of $P_i$ that are also in $P_j$. Similarly, let $v_{j,f}$ and $v_{j,l}$ be the first and the last vertices of $P_j$ that are also in $P_i$. 
We let $P^i_{f,l}$ denote the subpath of $P_i$ that starts at $v_{i,f}$ and ends at $v_{i,l}$. 
We let $P^j_{f,l}$ denote the subpath of $P_j$ that starts at $v_{j,f}$ and ends at $v_{j,l}$. 
An edge in a pair of intersecting paths $P_i, P_j$ is said to be \emph{exclusive} if it belongs to either $P^i_{f,l}$ or $P^j_{f,l}$ but not both.

With the above in mind, we describe the merging process. While there exists an edge in the path system that is exclusive in some pair of intersecting paths, we look for the edge that is exclusive in the smallest number of intersecting path pairs (this number is called the \emph{exclusivity frequency}). If such an edge does not exist, this implies that the path system is merged. Once we have selected an edge, call it $e$, we pick an arbitrary pair of intersecting paths where the edge $e$ is exclusive, and we proceed to merge this pair. Let the selected pair be $P_i$ and $P_j$. We can either reroute $P^i_{f,l}$ through $P^j_{f,l}$, or we can reroute $P^j_{f,l}$ through $P^i_{f,l}$. If one of the reroutings decreases the weight of the path system, it is chosen, and the path system is updated accordingly. Otherwise, we make both paths go through whichever of $P^i_{f,l}$ or $P^j_{f,l}$ maximizes token isolation. If rerouting through either of those subpaths isolates the same number of tokens, we reroute both paths through the subpath that does not contain the edge $e$ we selected initially. 

By definition of a merged path system, termination implies correctness. Therefore, it remains to show that the algorithm terminates. Merging two paths has one of three consequences:

\begin{enumerate}
    \item The decrease of the overall weight of the path system, a possible increase in the number of isolated tokens in the path system (because path merging may decrease the number of distinct edges used in the path system), and an indeterminate effect on the exclusivity frequencies of edges in unmerged path pairs
    \item The increase of the overall number of isolated tokens in the path system and an indeterminate effect on the exclusivity frequencies of edges in unmerged path pairs (no increase in the overall weight of the path system)
    \item The decrease of the exclusivity frequency of the edge that is exclusive in the smallest number of unmerged path pairs and an indeterminate effect on the exclusivity frequencies of edges in other unmerged path pairs (no increase in the overall weight of the path system or decrease in the number of isolated tokens)
\end{enumerate}

The first two consequences occur polynomially many times; we can isolate at most $\mathcal{O}(n)$ tokens, and since we assume that $\sum_{e \in E(G)}{w(e)} = \mathcal{O}(n^{c})$, the first consequence can occur at most $\mathcal{O}(n^{c + 1})$ times (since  $\sum_{P \in \mathcal{P}}{w(P)} = \mathcal{O}(n^{c + 1})$). We need to take into account the interaction between the first two consequences and the third consequence. The third consequence can occur $\mathcal{O}(nm)$ times in a row, as each edge can belong to all paths, of which we have $\mathcal{O}(n)$. Interleaving the third consequence with the first two consequences, both of which have an indeterminate effect on exclusivity frequencies, leads to a total of $\mathcal{O}(nm \cdot (n + n^{c + 1})) = \mathcal{O}(n^{c+2}m)$ path pair merges. Each path pair merge is executed in time $\mathcal{O}(m)$ and is preceded by a lookup for the edge that is exclusive in the smallest number of intersecting path pairs. This lookup is done in time $\mathcal{O}(n^2m)$, as it requires checking every path pair and iterating over the edges of the paths in each path pair. The overall running time of the path merging procedure is therefore $\mathcal{O}(n^{c+2}m \cdot (m + n^2m)) = \mathcal{O}(n^{c+4}m^2)$.

We still have to show that, once an edge is no longer exclusive in any unmerged path pair, that its exclusivity frequency can no longer increase as a result of the third consequence. 
If some merge increases the exclusivity frequency of the edge in question, since merging involves reusing edges that are already part of the path system, this implies that the edge already occurred exclusively in some unmerged path pair involving the path that the merging rerouted through. This contradicts the fact that the edge is no longer exclusive in any unmerged path pair, as needed. 
\end{proof}

\subsection{Step 2 -- Unwrap path system}\label{app:ups}

The second step of the aro subroutine is to compute an \emph{unwrapped path system} (UPS). An unwrapped path system is a MPS such that no path within it contains another path. A path $P_i$ is said to \emph{contain} a path $P_j$ if the intersection between $P_i$ and $P_j$ is $P_j$. 

\begin{lemma}[Unwrapped path system]\label{lemma:ups}
Given an $n$-vertex (positive) edge-weighted graph $G$ with $|E(G)| = m$, two sets $S,T \subseteq V(G)$, and a $T$-valid merged path system $\mathcal{P}$, we can compute, in time $\mathcal{O}(n^2m)$, a valid unwrapped path system $\mathcal{P'}$ such that $w(\mathcal{P'}) \leq w(\mathcal{P})$. Moreover, the number of distinct edges used in $\mathcal{P'}$ is at most the number of distinct edges used in $\mathcal{P}$.
\end{lemma}

\begin{proof}
We process the paths in an arbitrary order, and for every path $P$, we unwrap all the paths that it wraps. To do so, we go through the path system and we detect all paths whose source vertex and target vertex are both contained within the selected path $P$. Let $l$ be the number of such paths. We then separately sort the $l$ source vertices and the $l$ target vertices by their order of appearance within the selected path $P$. The final step assigns source $i$ to target $i$ in the ordering, and it does so for all $i$ in $\mathbb{N}^{+}([1,l])$; the path with source $i$ as source vertex is then rerouted to target vertex $i$ via the selected path $P$.

For a selected path $P$, this process destroys all wrappings within it. Assume it does not, that is, assume the assignment of sources to targets in order of appearance in $P$ fails to unwrap a wrapping. Without loss of generality, we will suppose that $P_g$ wraps $P_h$, and we will use $g$ and $h$ to designate the indices of the source and target vertices of $P_g$ and $P_h$ respectively after sorting within $P$, rather than them just being arbitrary indices for the tangled paths. The concerned vertices $v_{s_g}, v_{s_h}, v_{t_g}, v_{t_h}$, must have appeared in one of four orders. Two of those four orders will be discussed below, as the analysis for the other two is symmetrical. If the sequence of vertices in the initial selected path $P$ takes on the form $\ldots, v_{s_g}, \ldots, v_{s_h}, \ldots, v_{t_h}, \ldots, v_{t_g}, \ldots$, we have a contradiction; since $v_{s_g}$ appears before $v_{s_h}$ in $P$, $g < h$. Likewise, since $v_{t_h}$ appears before $v_{t_g}$ in $P$, $h < g$. The same is true if we have the form $\ldots, v_{s_g}, \ldots, v_{t_h}, \ldots, v_{s_h}, \ldots, v_{t_g}, \ldots$. The last two forms, which are identical to the two forms we presented but with the positions of $v_{s_g}$ and $v_{t_g}$ switched, yield the same contradiction. 

We still have to show that path unwrapping does not ``unmerge'' a path system, and that there are no wrappings left when the algorithm terminates. 

We start by showing that applying path unwrapping to a merged path system does not undo merging. It is sufficient to show that unwrapping paths within an arbitrary path in a merged path system does not give rise to a pair of paths that have more than one intersection. When unwrapping paths within a path, some paths are shortened and some paths are extended. Shortening a path does not create an additional intersection between it and any other path, so we only have to worry about the paths that get extended as a result of the unwrapping. If the extension of some path $P_x$ makes it intersect some other path $P_y$ more than once, the selected arbitrary path $P_z$ which initially wrapped the now-extended path $P'_x$ intersects $P_y$ more than once, because $P'_x$ is a subpath of $P_z$, which contradicts the assumption that we started with a merged path system.

Finally, we show that no unwrapped path remains after path unwrapping terminates. Assume that path $P_i$ remains wrapped in path $P_j$ after termination. We know that, in its execution, the algorithm should have processed $P_j$ and all the paths that contain it. We proved that unwrapping paths within any of the paths that contain $P_j$ (including $P_j$ itself) would eliminate the wrapping of $P_i$ within $P_j$. Since the wrapping persisted, it has to be the case that it was caused by the unwrapping of paths within another path that does not contain $P_j$. This is not possible, as the only paths that modify $P_i$ and $P_j$ via unwrapping are paths that contain them. 

The algorithm unwraps paths within every path; unwrapping paths within a single path can be executed in time $\mathcal{O}(nm)$. Finding the wrapped paths can be accomplished in time $\mathcal{O}(m)$, whereas reconstructing the wrapped paths is done in time $\mathcal{O}(nm)$, which is equal to the sum of their lengths, and since there are $\mathcal{O}(n)$ paths in total, the running time of the path unwrapping procedure is $\mathcal{O}(n^2m)$.
\end{proof}

\subsection{Step 3 -- Detect and break cycles}\label{app:cps}

The third step of the aro algorithm detects and breaks cycles in a path system to compute a \emph{cycle-free path system} (CPS), which is a UPS such that the graph it induces is a forest. 

We use $G[\mathcal{P}]$ to denote the graph induced by (the vertices of the paths of) a path system $\mathcal{P}$, which we also call the \emph{path system graph}. 
A cycle is either represented by a sequence of vertices $\langle v_1, v_2, \ldots, v_k \rangle$ or a sequence of edges $\langle e_1, e_2, \ldots, e_k \rangle$. Given a path system $\mathcal{P} = \{P_1, \ldots, P_{k}\}$ that induces a cycle $C$ characterized by its edge set $\mathcal{E}$, we define an edge coloring of the cycle as a function $col_C : \mathcal{E} \mapsto \{1, \ldots, k\}$ (the subscript indicating the cycle is dropped if the cycle we are referring to is clear from the context), where color $i$ is associated with $P_i$. An edge coloring of the path system $\mathcal{P}$ (or, more broadly, a set of paths) is defined analogously. If $col(e_i) = j$, we say that the edge $e_i$ is \emph{$j$-colored}, and $e_i$ is \emph{$j$-colorable} if and only if it is on the path $P_j$. We say that a path is $j$-colorable if all its edges are $j$-colorable. Note that even though edges can appear in more than one path (which implies that an edge can be colored using one of multiple colors), we are interested in a special type of edge colorings, the purpose of which will become clearer later. 
A cycle is \emph{contiguously colored} if any two edges $e_1$, $e_2$ that have the same color $j$ are separated by a sequence of edges along the cycle that are $j$-colored.
%
%A cycle is \emph{contiguously colored} if there does not exist a pair of edges $e_i$ and $e_j$ in its edge representation such that $col(e_i) = col(e_j)$ with the two edges not being contiguous in the cycle. 
%
A cycle that is not contiguously colored is \emph{discontiguously colored}. A cycle is \emph{contiguously colorable} if there exists a coloring of its edges that makes it contiguously colored. 
%Contiguously colored and continuously colorable paths are defined analogously. 
A color in a cycle is \emph{discontiguous} if there exist two non-consecutive edges $e$ and $e'$ in the cycle such that $col(e) = col(e')$ and neither of the two subpaths that connect them along the cycle are $col(e)$-colored.

\begin{theorem}\label{thm:cps}
Given an $n$-vertex (positive) edge-weighted graph $G$ with $|E(G)| = m$ and $\sum_{e \in E(G)}{w(e)} = \mathcal{O}(n^{c})$ for some positive integer $c$, two sets $S,T \subseteq V(G)$, and a $T$-valid unwrapped path system $\mathcal{P}$, we can compute, in time $\mathcal{O}(n^{c+6}m)$, a valid cycle-free path system $\mathcal{P'}$ such that $w(\mathcal{P'}) \leq w(\mathcal{P})$. Moreover, the number of distinct edges used in $\mathcal{P'}$ is at most the number of distinct edges used in $\mathcal{P}$.
\end{theorem}

The proof of the theorem is quite involved, so we break it into several parts. We first prove that whenever the path system graph $G[\mathcal{P}]$ contains a cycle passing through some arbitrary edge $e$ then $G[\mathcal{P}]$ must contain a \emph{special cycle} (defined later) passing through $e$. 
We then describe a procedure to find a special cycle passing through $e$. We further describe our approach to break the special cycle, and finally prove termination of the cycle-breaking procedure~(Lemma~\ref{lemma:cycle_breaking_termination}). While our results are applicable for arbitrary edges, we apply the algorithms we derive from the lemmas to a specific edge $e^\star$; the careful selection of the edge $e^\star$ is what guarantees termination. 

Given a path system, we let the \emph{frequency} of an edge denote the number of paths containing it. Cycle detection consists of finding whether there is a cycle in $G[\mathcal{P}]$. Finding a cycle can be achieved using any graph traversal algorithm; either a breath-first search (BFS) or a depth-first search (DFS) is sufficient. We wish to obtain additional information; if a cycle is found, we look for the edge $e^\star$, which is the edge with the smallest frequency among all edges contained in cycles.

To find the edge of interest, i.e., $e^\star$, we sort the edges with non-zero frequency in non-decreasing order of frequency, and then, in this ordering, we look for the earliest edge that is part of a cycle. For cycle detection, let $u$ and $w$ be the endpoints of an arbitrary edge $e$. Then, $e$ is part of a cycle in $G[\mathcal{P}]$ if and only if $w$ is reachable from $u$ in $G[\mathcal{P}] - e$, where $G[\mathcal{P}] - e$ denotes the graph obtained from $G[\mathcal{P}]$ after deleting the edge $e$. 
After the edge $e^\star$ is found, we look for paths that induce a special cycle, which is a cycle that contains $e^\star$ and has particular properties.

Before describing the procedure that allows us to identify the desired set of paths, we provide a few additional relevant definitions. We define an \emph{$e$-path} as a path that contains the edge $e$. The results from Lemmas~\ref{cycle_breaking_contiguous_colors},~\ref{cycle_breaking_two_cycle_edge_paths}, and~\ref{single_induced_cycle_existence} show that if there is a cycle in the path system graph that passes through an arbitrary edge $e$, then there must exist a \emph{special cycle} that passes through edge $e$. A special cycle that passes through edge $e$ is a cycle that:

\begin{enumerate}
    \item is contiguously colorable
    \item is induced by a set of paths that:
    \begin{enumerate}
        \item is inclusion-minimal
        \item contains at most two $e$-paths
        \item induces a single cycle, i.e., it induces no cycle other than the special cycle itself
    \end{enumerate}
\end{enumerate}

% The correctness of our procedure for finding special cycles follows from Lemmas~\ref{cycle_breaking_contiguous_colors},~\ref{cycle_breaking_two_cycle_edge_paths} and~\ref{single_induced_cycle_existence} as well.

\subsubsection{Proof of existence of special cycles}\label{app:special-cycle}

\begin{lemma}
\label{cycle_breaking_contiguous_colors}
If there is a cycle $C$ in $G[\mathcal{P}]$ that passes through an arbitrary edge $e$, then there is a contiguously colorable cycle $C'$ in $G[\mathcal{P'}]$, $\mathcal{P'} \subseteq \mathcal{P}$, that also passes through edge $e$. 
\end{lemma}

\begin{proof}
Consider any cycle $C$ that passes through an arbitrary edge $e$. If the cycle is contiguously colored, there is nothing to prove. Otherwise, we describe how the existence of this cycle implies the existence of a contiguously colorable cycle that includes $e$. We call a \emph{$j$-colored segment} a maximal contiguous sequence of $j$-colored edges in the edge representation of the cycle. We consider an edge coloring of $C$, $col_C$, that has a minimal number of segments in the cycle; this guarantees that the edges that separate $j$-colored segments on the cycle and that are not $j$-colored themselves are not all $j$-colorable because, otherwise, we can reduce the number of segments in the cycle.

We first cover notation we will be recurrently using in this proof. For any two vertices $u, v \in V(P_k)$, where $P_k$ is an arbitrary path, we denote the subpath of $P_k$ going from $u$ to $v$ by $P_{k, u \rightarrow v}$. For any cycle $C_k$, any edge $\eta \in E(C_k)$ and any two vertices $u, v \in V(C_k)$, we use $P_{C_k, u \rightarrow \eta \rightarrow v}$ (resp. $P_{C_k,u \rightarrow v}$) to refer to the subpath between $u$ and $v$ along cycle $C_k$ that goes through (resp. does not go through) $\eta$.

We describe a procedure that reduces the number of segments in the cycle. Let $col_C(e)$ be the color of the edge $e = \{u, w\}$. We first handle making all other colors in $C$ contiguous, which is a process we call \emph{peripheral color merging}. We then handle making the color of $e$ in the resulting cycle contiguous, which we refer to as \emph{edge color merging}.

\textbf{Peripheral color merge.} Peripheral color merging merges multiple $j$-colored segments into a single $j$-colored segment, where $j \neq col_C(e)$. Consider the vertex representation $v_1, v_2, \ldots, v_k$ of the cycle (where $v_1 = u$ and $v_k = w$ are the endpoints of $e$). Let $v_{j_p}$ and $v_{j_q}$ be the earliest and the latest vertices in the vertex representation of the cycle belonging to path $P_j$ ($v_{j_p}$ and $v_{j_q}$ may be $v_1$ and $v_k$, respectively). Clearly, $P_{j, v_{j_p} \rightarrow v_{j_q}}$ is $j$-colorable. We consider two cases: 
\begin{enumerate}
    \item \emph{$P_{j, v_{j_p} \rightarrow v_{j_q}}$ does not contain $e$:} $P_{j, v_{j_p} \rightarrow v_{j_q}}$ can replace $P_{C, v_{j_p} \rightarrow v_{j_q}}$ in the vertex sequence of the cycle, and we set the color of all the edges in $P_{j, v_{j_p} \rightarrow v_{j_q}}$ to $j$. It is easy to see that $P_{j, v_{j_p} \rightarrow v_{j_q}}$ and $P_{C, v_{j_p} \rightarrow e \rightarrow v_{j_q}}$ are internally vertex-disjoint, so the new vertex sequence is that of a cycle containing the edge $e$. The original cycle was therefore transformed into a new cycle $C'$ where color $j$ is contiguous, as ensured by the definition of $v_{j_p}$ and $v_{j_q}$. This implies that the cycle transformation reduced the number of segments by at least 1, since there were at least two $j$-colored segments in $C$. 
    \item \emph{$P_{j, v_{j_p} \rightarrow v_{j_q}}$ contains $e$:} We look at the vertices in $V(C) \cap V(P_j)$, and we partition them into two sets; the set of vertices $V_{j_b}$ that are before $e$ in the vertex representation of $P_j$, and the set of vertices $V_{j_a}$ that are after $e$ in the vertex representation of $P_j$. It is evident that $V_{j_b} \cap V_{j_a} = \emptyset$ and that $|V_{j_b}| \neq 0$, $|V_{j_a}| \neq 0$, because both $v_{j_p}$ and $v_{j_q}$ belong to one of the two sets, and they belong to different sets. Next, we search for two vertices $v_{j_b}$ and $v_{j_a}$ in the vertex representation of the cycle, such that $v_{j_b} \in V_{j_b}$, $v_{j_a} \in V_{j_a}$, and none of the internal vertices of $P_{C, v_{j_b} \rightarrow v_{j_a}}$ are in $V_{j_a}$ or $V_{j_b}$. There exist two paths between $v_{j_b}$ and $v_{j_a}$, $P_{j, v_{j_b} \rightarrow v_{j_a}}$ and $P_{C, v_{j_b} \rightarrow v_{j_a}}$, and those two paths are internally vertex-disjoint because $P_{C, v_{j_b} \rightarrow v_{j_a}}$ sharing a vertex with $P_{j, v_{j_b} \rightarrow v_{j_a}}$ that is not $v_{j_b}$ or $v_{j_a}$ would contradict our choice of $v_{j_b}$ and $v_{j_a}$. Those two paths will make up our transformed cycle $C'$. The former of the two paths will be $j$-colored (meaning that we change the color of $e$ as well), whereas the coloring of the latter of the two paths is unchanged. Color $j$ is contiguous in $C'$, given how we defined $v_{j_b}$ and $v_{j_a}$, meaning that we have reduced the number of segments by at least 1. Notice that $col_{C'}(e) \neq col_C(e)$.
\end{enumerate}

We exhaustively apply the procedure above, which clearly terminates, since the number of segments is reduced after every step, and we end up with a cycle $C''$ with edge coloring $col_{C''}$, where the only discontiguous color, if any, is $col_{C''}(e)$, the color of the edge $e$ (note that $col_{C''}(e)$ may or may not be equal to $col_C(e)$). In what follows, we use $\rchi$ (resp. $P_{\rchi}$) to refer to $col_{C''}(e)$ (resp. the path associated with $col_{C''}(e)$). We now describe how $\rchi$ can be made contiguous. 

\begin{figure}[t]
    \centering
    \includegraphics[]{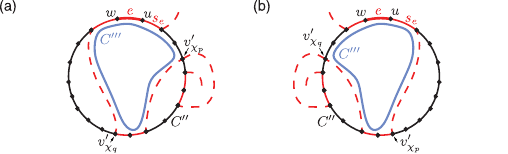}
    \caption{
    \textbf{Merging the edge color for the case where exactly one of $V_{\rchi_b} \setminus V(s_e)$ and $V_{\rchi_a} \setminus V(s_e)$ is empty.}
    (a-b)~Starting from cycle $C''$ where the edge color (red) is discontiguous, we construct a cycle $C'''$ (blue line) in which the color of edge $e$ (red) is contiguous via path $P_{\rchi}$ (red solid line on the cycle to highlight the segments, red dashed line indicates edges in the path that are not shared with the cycle).
    }
    \label{fig:coloring_cases}
\end{figure}

\textbf{Edge color merge.} Edge color merging merges multiple $\rchi$-colored segments in $C''$ into a single $\rchi$-colored segment. In the vertex representation of $C''$, we are interested in segments that are $\rchi$-colored. If there is only one $\rchi$-colored segment, the cycle is contiguously colored and we are done. Otherwise, we have two or more $\rchi$-colored segments. We can assume that the $\rchi$-colored segments are maximal, that is, none of the $\rchi$-colored segments can be extended by recoloring edges in $C''$ to include more edges. We let $s_e$ be the segment containing edge $e$ and we let $V(s_e)$ denote the vertices of $s_e$. We define $V_{\rchi_b}$, $V_{\rchi_a}$, $v_{\rchi_b}$ and $v_{\rchi_a}$ analogously to $V_{j_b}$, $V_{j_a}$, $v_{j_b}$, and $v_{j_a}$, respectively. Since there are two or more $\rchi$-colored segments in $C''$, it must be the case that at least one of $V_{\rchi_b} \setminus V(s_e)$ or $V_{\rchi_a} \setminus V(s_e)$ is nonempty, otherwise, $(V_{\rchi_b} \cup V_{\rchi_a}) \setminus V(s_e) = \emptyset$, i.e., $V(P_{\rchi}) \cap V(C) = V(s_e)$, which would imply that we have a single $\rchi$-colored segment in the cycle. The remaining two cases are considered below:

\begin{enumerate}
    \item \emph{$V_{\rchi_b} \setminus V(s_e)$ and $V_{\rchi_a} \setminus V(s_e)$ are both nonempty:} We apply the algorithm described in the second case of peripheral color merging in order to merge color $\rchi$ (i.e., $P_{\rchi}$ plays the role of $P_j$ in that case). What makes applying said algorithm possible here is the fact that both $V_{\rchi_b} \setminus V(s_e)$ and $V_{\rchi_a} \setminus V(s_e)$ are nonempty, which implies that $v_{\rchi_b}$ and $v_{\rchi_a}$ both exist, and the construction of the cycle follows from that. The resulting cycle $C'''$ is a cycle where color $\rchi$ is contiguous, so we are done.
    \item \emph{Exactly one of $V_{\rchi_b} \setminus V(s_e)$ or $V_{\rchi_a} \setminus V(s_e)$ is empty:} We look at $v'_{\rchi_p}$ and $v'_{\rchi_q}$, which are the earliest and the latest vertex, respectively, in the vertex representation of the cycle belonging to path $P_{\rchi}$ but not segment $s_e$. We either have $v'_{\rchi_p}, v'_{\rchi_q} \in V_{\rchi_b} \setminus V(s_e)$, or $v'_{\rchi_p}, v'_{\rchi_q} \in V_{\rchi_a} \setminus V(s_e)$. In the vertex representation of $P_{\rchi}$, either $u$ comes before $w$, or $w$ comes before $u$. We have four different subcases we split into two different groups: 
    \begin{enumerate}
        \item $v'_{\rchi_p}, v'_{\rchi_q} \in V_{\rchi_b} \setminus V(s_e)$ and $w$ comes before $u$ in the vertex representation of $P_{\rchi}$, or $v'_{\rchi_p}, v'_{\rchi_q} \in V_{\rchi_a} \setminus V(s_e)$ and $w$ comes after $u$ in the vertex representation of $P_{\rchi}$ (Fig.~\ref{fig:coloring_cases}a): There exist two vertex-disjoint paths between $w$ and $v'_{\rchi_p}$; $P_{\rchi, w \rightarrow v'_{\rchi_p}}$  and the path $w, u,\ldots,v'_{\rchi_p}$ along the cycle going through $e$. Those two paths will make up our transformed cycle $C'''$. The former of the two paths will be $\rchi$-colored, whereas the coloring of the latter of the two paths is unchanged. Color $\rchi$ is contiguous in $C'''$, and no color was made discontiguous, as needed.
        \item $v'_{\rchi_p}, v'_{\rchi_q} \in V_{\rchi_b} \setminus V(s_e)$ and $u$ comes before $w$ in the vertex representation of $P_{\rchi}$, or $v'_{\rchi_p}, v'_{\rchi_q} \in V_{\rchi_a} \setminus V(s_e)$ and $u$ comes after $w$ in the vertex representation of $P_{\rchi}$ (Fig.~\ref{fig:coloring_cases}b): There exist two vertex-disjoint paths between $u$ and $v'_{\rchi_q}$; $P_{\rchi, u \rightarrow v'_{\rchi_q}}$  and the path $u, w,\ldots,v'_{\rchi_q}$ along the cycle (in reverse order) going through $e$. Those two paths will make up our transformed cycle $C'''$. The former of the two paths will be $\rchi$-colored, whereas the coloring of the latter of the two paths is unchanged. Color $\rchi$ is contiguous in $C'''$, and no color was made discontiguous, as needed.
    \end{enumerate}
\end{enumerate}
This completes the proof. 
\end{proof}

\begin{lemma}
\label{cycle_breaking_two_cycle_edge_paths}
Let $C$ be a %contiguously colored 
cycle in $G[\mathcal{P}]$ that passes through an arbitrary edge $e$ and let $\{P_1, \ldots, P_k\}$ be an inclusion-minimal set of paths from $\mathcal{P}$ that contains all edges of $C$. If there exist more than two $e$-paths in $\{P_1, \ldots, P_k\}$ then there exists a contiguously colored cycle $C'$ in $G[\mathcal{P}]$ that passes through edge $e$ such that an inclusion-minimal set of paths $\{P'_1, \ldots, P'_{k'}\} \subset \{P_1, \ldots, P_k\} \subseteq \mathcal{P}$ that contains all edges of $C'$ has at most two $e$-paths (and $k' < k$). 
\end{lemma}

\begin{proof}
We can assume that the edges of the same color are contiguous along $C$, by Lemma~\ref{cycle_breaking_contiguous_colors}. We show how to transform $C$ into a cycle that uses two or fewer $e$-paths.

If $C$ already has this property, there is nothing to prove. Otherwise, $C$ uses three or more $e$-paths and we describe how $C$ can be transformed. In $C$, we are dealing with a total of $\ell$ different $e$-paths, having colors $1$ to $\ell$, with color $1$ being the color of the edge $e$.
We now describe a process that yields a cycle with the desired properties. We iterate through the $\ell$ possible colors for the edge $e$ in order, starting from color $2$, and, for every color, we change the color of $e$ accordingly and we merge the edge $e$ with the other segment of the same color. The first color change may cause up to two discontiguities; one discontiguity in color $1$ if $e$ is not at the extremity of the $1$-colored segment, and one discontiguity in color $2$ if the $2$-colored segment does not share an endpoint of $e$ with the $1$-colored segment. If the color change makes color $2$ discontiguous, we fix the discontiguity in color $2$ using the construction from case 2 of the edge color merge subroutine from Lemma~\ref{cycle_breaking_contiguous_colors}. This construction can be applied here, despite the fact that its precondition may not be satisfied, because the vertices in the $2$-colored segment that does not contain $e$ are either all in $V_{2_b}$ or $V_{2_a}$, and the other end of the path whose vertices are in the cycle but are not part of a $2$-colored segment can be safely ignored. This construction restores color 2's contiguity, and eliminates the discontiguity in color 1 (if any) by either discarding or recoloring one of the two $1$-colored segments. Note that this process may eliminate some colors in the cycle; if an $e$-path color is eliminated, it is skipped in the process of iterating through the colors for the edge $e$. Next, we color the edge $e$ using color $3$. Note that this does not cause a discontiguity in color $2$ because the edge $e$ is at the extremity of color $2$'s segment, owing to the construction from the edge color merging subroutine; we merge color $3$, we color the edge $e$ using color $4$, and so on. Any coloring and merging can eliminate a color, but can never cause a discontiguity in any of the other colors. We then loop back to color $1$. At this point, there are no color discontiguities in the cycle, and no color change applied on the edge $e$ can cause color discontiguities. We claim that the cycle resulting from this process conforms to the property we are looking for. That is, it uses two or fewer $e$-paths.

Suppose it does not, that is, suppose that it uses three or more $e$-paths. If any of those $e$-paths is not incident to the edge $e$, the color of the edge $e$ can be changed to cause a discontiguity, which contradicts what we said earlier about changes in the color of $e$ involving an $e$-path color not being able to cause discontiguities. Otherwise, the only case where we have three or more contiguous segments incident to $e$ is when $e$ and the two edges that are adjacent to it along the cycle have three different colors. However, in that case, $e$ can be colored using one of the other two colors and the path whose color was previously used to color $e$ can be discarded, contradicting minimality.
\end{proof}

\begin{lemma}
\label{single_induced_cycle_existence}
If there exists a cycle $C$ in $G[\mathcal{P}]$ containing edge $e$, then there exists an inclusion-minimal set of paths $\mathcal{P}^\star = \{P_1, \ldots, P_{k}\} \subseteq \mathcal{P}$, $k \geq 3$, that induces a unique cycle $C^\star$ such that $C^\star$ contains $e$. 
\end{lemma}

\begin{proof}
Assume the set of paths $\{P_1, \ldots, P_{k}\}$ is inclusion-minimal and induces multiple cycles. Let $\mathcal{C} = \{C_1, C_2, \ldots, C_q\}$ denote the set of all cycles in $G[\mathcal{P}^\star]$. Let $\mathcal{C}_e \subseteq \mathcal{C}$ denote the set of cycles containing $e$ and let $\mathcal{C}_{\bar{e}} = \mathcal{C} \setminus \mathcal{C}_e$ denote all remaining cycles. Two colors $i$ and $j$ are said to be \emph{adjacent} on some cycle if some coloring of the cycle includes two (consecutive) segments colored $i$ and $j$ that share at least one vertex. Clearly, if two colors are adjacent, their corresponding paths intersect (either on a single vertex or on multiple consecutive vertices). 

By Lemma~\ref{cycle_breaking_two_cycle_edge_paths}, we know that $\mathcal{P}^\star$ contains at most two $e$-paths; as otherwise we can reduce the number of paths and maintain a cycle containing $e$, contradicting the minimality of $\mathcal{P}^\star$. Moreover, we know, by Lemma~\ref{cycle_breaking_contiguous_colors}, that at least one cycle in $\mathcal{C}_{e}$ can be contiguously colored. We denote this cycle by $C'$. Every path in $\{P_1, \ldots P_{k}\}$ contains at least one edge in every cycle of $\mathcal{C}_{e}$ that does not appear in any other path; otherwise, a path can be omitted from $\mathcal{P}^\star$ and the remaining paths will still induce a cycle containing $e$, contradicting the minimality of $\mathcal{P}^\star$. We call such edges~\emph{private edges}, i.e., edges belonging to cycles in $\mathcal{C}$ and to a single path in $\mathcal{P}^\star$.  
Without loss of generality, we let $\langle 1, 2, \ldots, k \rangle$ denote the color ordering of the segments of $C'$ (viewed in the clockwise order and assuming the segment containing $e$ is colored $1$). 

Now, assume that $\mathcal{P}^\star$ contains three or more $e'$-paths for some $e'$ appearing in any cycle of $\mathcal{C}$. Then, by Lemma~\ref{cycle_breaking_two_cycle_edge_paths}, we can find a proper subset of $\mathcal{P}^\star$, say $\mathcal{P}'$, such that there exists a contiguously colored cycle $C''$ in $G[\mathcal{P}']$; $\mathcal{P}'$ 
contains all edges of $C''$, has at most two $e'$-paths, and hence at least one path less than $\mathcal{P}^\star$, i.e., $|\mathcal{P}'| < |\mathcal{P}^\star|$. This implies that there exist two paths $P_i$ and $P_j$, where $i \neq j$ and $|i - j| \geq 2$, that are intersecting in $C''$ (have consecutive segments in $C''$) but whose segments are not consecutive in $C'$. First, assume that neither $P_i = P_1$ nor $P_j = P_1$. Then, we can construct a cycle containing $e$ and using at most $k - 1$ colors (paths), a contradiction. For the remaining case, we assume, without loss of generality, that $P_j = P_1$. Depending on whether the intersection of $P_i$ and $P_j$ includes $e$, is to the left of $e$ (in the vertex ordering of $P_1$), or is to the right of $e$, we can still construct a cycle containing $e$ and using at most $k - 1$ paths, again contracting our choice of $\mathcal{P}^\star$. Putting it all together, we can now assume that every edge of $\mathcal{P}^\star$ belonging to a cycle of $\mathcal{C}$ can appear in at most two paths. Using a similar argument, we also conclude that every cycle in $\mathcal{C}$ must contain at least one segment from each path in $\{P_1, \ldots, P_{k}\}$ (along with a private edge); otherwise we can again construct a cycle containing $e$ and using fewer paths. 

It remains to show that $\mathcal{C} = \{C^\star\}$, i.e., $G[\mathcal{P}^\star]$ contains a unique cycle $C^\star$. Assume otherwise and let $C'$ denote the contiguously colored cycle containing $e$. Let $C''$ be any other cycle. We have shown that $C''$ must contain at least one segment of each color and no edge of the path system belongs to more than two paths. There are two cases to consider. If $C''$ can be contiguously colored such that the ordering matches that of $C'$, i.e., $\langle 1, 2, \ldots, k \rangle$ with edge $e$ colored $1$, then we claim that $C' = C''$. To see why, consider two paths $P_{i}$ and $P_{i+1}$ whose colors are consecutive in the color ordering. Both $C'$ and $C''$ have to go through at least one private edge in $P_{i}$ and one private edge in $P_{i+1}$, and since the colors of $P_{i}$ and $P_{i+!}$ are contiguous in the color ordering, the cycles must go through the private edges in question via the unique intersection of the two consecutive path segments (using either a single vertex of the intersection or all edges of the intersection). Since the pair of paths whose colors are consecutive in the color ordering were picked arbitrarily, it follows that both $C'$ and $C''$ contain the edges (or vertices) in all the path intersections. A similar argument can be used to show that the inclusion of the edges (or vertices) in the path intersections in both cycles implies the inclusion of the edges in the subpaths between the path intersections in both $C'$ and $C''$. Since the two cycles have the same edge set, it follows that they are in fact the same cycle, as needed.

Finally, assume that $C''$ cannot be contiguously colored so that its color ordering matches that of $C'$. If for some edge coloring of $C''$, we get a color $i$ that is adjacent to at least $3$ distinct colors in $C'$ and $C''$, then we denote those colors by $j$, $l$, and $p$. One of the colors $j$, $l$, or $p$ will not be adjacent to color $i$ in the color ordering of $C$. Without loss of generality, assume that color $j$ is not adjacent to color $i$ in the color ordering of $C'$. Starting from the color ordering of $C'$, we can construct a sequence of segments that omits the colors between $i$ and $j$ in the color ordering of $C'$ because paths $P_{i}$ and $P_{j}$ intersect (given that colors $i$ and $j$ are adjacent). The constructed sequence of colors corresponds to a set of paths with cardinality strictly less than $k$ that induces a cycle (and contains $e$), thus contradicting the fact that the set of paths $\{P_1, \ldots, P_{k}\}$ was assumed to be inclusion-minimal with regards to $C^\star$. Moreover, we can assume a coloring of $C''$ that minimizes the number of segments which implies that the color ordering of $C''$ cannot contain a contiguous sequence of the form $i$, $i+1$, $i$, since this implies that $P_i$ and $P_{i + 1}$ intersect at least twice, another contradiction. It follows that the coloring of $C''$ must be contiguous and ordered as in $C'$. As noted in the previous case, this implies that $C' = C''$ and we are done. 
\end{proof}

The combination of the three previous lemmas allows us to reduce the problem of finding cycles to the problem of finding special cycles. The reason why special cycles are of interest to us is because they make the cycle-breaking algorithm work, as the cycle-breaking procedure is specifically designed to break special cycles.

\subsubsection{Procedure to find a special cycle}\label{app:find-cycles}

We now describe the procedure for finding a special cycle containing the edge $e^\star$ and obtain the set of paths that induce this special cycle. Recall that $e^\star$ is the edge with the smallest frequency among all edges that are contained in cycles.

Let $z$ be the number of $e^\star$-paths in the path system. For each $e^\star$-path, we construct $2(z-1) + 1$ different graphs that we call \emph{path intersection graphs}, for a total of $z(2(z-1) + 1)$ path intersection graphs. Every path intersection graph is built out of a selection of one $e^\star$-path as a \emph{base path}, and at most one other $e^\star$-path as a \emph{support path}, where a base path is an $e^\star$-path whose color we use for $e^\star$, and a support path is an $e^\star$-path that is needed as part of the special cycle we are looking for. In every path intersection graph, we add a vertex $v_{P_i}$ for every path $P_i$ that is not an $e^\star$-path, and we add an edge between two such vertices if the corresponding paths intersect. Moreover, in every path intersection graph with a path $P_b$ as its base path, we add two vertices $l_{P_b}$ and $r_{P_b}$. The vertex $l_{P_b}$ (resp. $r_{P_b}$) is associated with all the vertices in the base path that occur to the left (resp. to the right) of the edge $e^\star$ in the vertex sequence of the path (including one of the endpoints of the edge $e^\star$ in both cases). If some path $P_j$ that is not an $e^\star$-path intersects $P_b$, we add an edge from $v_{P_j}$ to either $l_{P_b}$ or $r_{P_b}$, depending on whether it intersects $P_b$ in a subpath associated with $l_{P_b}$ or $r_{P_b}$ (clearly, a path cannot intersect both subpaths without going through the edge $e^\star$ because the path system is merged). Support paths may either be treated as left-intersecting paths (i.e., the vertex corresponding to the support path is connected to $l_{P_b}$), right-intersecting paths (i.e., the vertex corresponding to the support path is connected to $r_{P_b}$), or neither, and that explains why we construct $2(z-1) + 1$ different path intersection graphs ($z-1$ path intersection graphs for the $z-1$ different choices of left-intersecting paths, $z-1$ path intersection graphs for the $z-1$ different choices of right-intersecting paths, and one path intersection graph with no support path) for every base path.

We claim that the problem of obtaining a set of paths that induces a special cycle containing edge $e^\star$ reduces to running $(l_{P_b}, r_{P_b})$-reachability queries on the constructed $z(2(z-1) + 1)$ path intersection graphs. More specifically, we use BFS to check whether $r_{P_b}$ is reachable from $l_{P_b}$, and if we find any yes-instance, we reconstruct the set of paths by using the BFS tree of the instance with the shortest path between $r_{P_b}$ and $l_{P_b}$; this set of paths induces a special cycle.

\begin{lemma}
\label{cycle_breaking_reduction}
Paths $P_1, P_2, \ldots, P_k$ induce a special cycle containing edge $e^\star$ $(e^\star \in E(P_1))$ if and only if $r_{P_1}$ is reachable from $l_{P_1}$ through vertices $v_{P_2}, v_{P_3}, \ldots, v_{P_k}$ in one of the generated reachability instances. 
\end{lemma}

\begin{proof}
We handle the forward direction first. If we have a special cycle, we know that it is contiguously colorable, that there exists a set of paths that induce it and induce no other cycles, and that it uses at most two $e^\star$-paths. We are able to extract a minimal sequence of paths $P_1, P_2, \ldots, P_k$ such that the only path intersections that exist within the sequence of paths are between two contiguous paths or the first and the last path in the sequence (if those are not the only path intersections, we end up with more than one induced cycle). There are two cases we need to handle:

\textbf{Case 1:} $P_1$ is the only $e^\star$-path in the special cycle; there is a path from $v_{P_2}$ to $v_{P_k}$ in all of the path intersection graphs that have $P_1$ as the base path, since $P_i$ intersects $P_{i + 1}$ for all $i \in \mathbb{N}^{+}([2,k-1])$. The path intersection graphs of interest to us are the ones with no support paths. $P_2$ and $P_k$ each intersect one side of $P_1$, otherwise, $e^\star$ is not contained in the special cycle, therefore, there exists a path intersection graph that either has edges from $v_{P_2}$ to $l_{P_1}$ and from $v_{P_k}$ to $r_{P_1}$, or edges from $v_{P_2}$ to $r_{P_1}$ and from $v_{P_k}$ to $l_{P_1}$; $r_{P_1}$ is reachable from $l_{P_1}$ either way. 

\textbf{Case 2:} we are dealing with two $e^\star$-paths in the special cycle, meaning that $P_1$ and exactly one of $P_2$ or $P_k$ are $e^\star$-paths (note that the second $e^\star$-path has to be one of those two paths, otherwise, we contradict the inclusion-minimality property of the set of paths $\{P_1, P_2, \ldots, P_k\}$). Here, $P_1$ acts as a base path, and exactly one of $P_2$ or $P_k$ acts as a support path. We therefore consider the two corresponding $(l_{P_1}, r_{P_1})$-reachability instances, and using an argument similar to the one we employed for the first case, at least one of those two reachability instances will be a yes-instance.
 
We now handle the backward direction. Assume $r_{P_1}$ is reachable from $l_{P_1}$, and let $l_{P_1}, v_{P_2}, v_{P_3}, \ldots, v_{P_k}, r_{P_1}$ be the corresponding shortest path that we obtained from the BFS tree. We can easily verify that the existence of such a path implies the existence of a subset of paths that induce cycles, this subset of paths being $P_1, P_2, \ldots, P_k$. 

We still have to show that $P_1, P_2, \ldots, P_k$ induce a special cycle. The easiest criterion to verify is that no more than two $e^\star$-paths are involved, as this follows by construction of the path intersection graphs, since all of them involve either one base path and no support path or one base path and one support path.

\begin{figure}[t]
    \centering
    \includegraphics[]{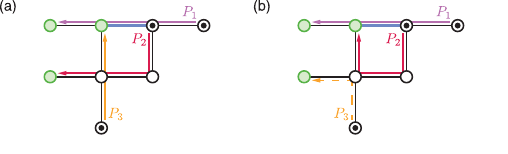}
    \caption{\textbf{Leveraging a source-target corner to reduce the number of paths that induce the special cycle.} (a)~Example of a sequence of three paths that induce a special cycle that contains the edge $e^\star$ (thick blue line).
    (b)~Updated sequence of paths after swapping the target vertices of $P_2$ and $P_3$, which can then be removed from the sequence of paths, as it no longer contributes any edge to the special cycle. 
    }
    \label{fig:cycle_breaking_st_corner}
\end{figure}
 
We now show that the path from $l_{P_1}$ to $r_{P_1}$ implies the corresponding paths induce a single cycle.

\textbf{Case 1:} a proper subset of the paths induces a cycle; we pick the smallest such subset. Let $\overleftarrow{P}$ be the path in said subset such that $v_{\overleftarrow{P}}$ is the earliest occurring vertex in the path from $l_{P_1}$ to $r_{P_1}$. Within the subset, there must exist two paths $\overleftarrow{P}$ intersects such that the vertices associated with those two paths occur later than $v_{\overleftarrow{P}}$ in the path from $l_{P_1}$ to $r_{P_1}$; let those paths be $P_j$ and $P_l$. Since $v_{\overleftarrow{P}}$ is the earliest occurring vertex, BFS visited it earlier than $v_{P_j}$ and $v_{P_l}$. Therefore, in the BFS tree, $v_{P_j}$ and $v_{P_l}$ are children of $v_{\overleftarrow{P}}$; the path from a leaf ($r_{P_1}$ specifically) to the root cannot therefore include both $v_{P_j}$ and $v_{P_l}$.

\textbf{Case 2:} the set of paths is inclusion-minimal and induces multiple cycles; we know this is not possible due to Lemma~\ref{single_induced_cycle_existence}.
The color contiguity criterion follows immediately, as we have already proven in Lemma~\ref{cycle_breaking_contiguous_colors} that any cycle can be turned into a contiguously colored cycle without adding new paths to the set of paths that induces it. 
%
% With that being said, we will prove something even stronger; all possible colorings of the cycle that the set of paths $P_1, P_2, \ldots, P_k$ induces are contiguous.
% 
% Suppose we were able to discontiguously color the cycle. Pick the discontiguous color that corresponds to the path whose vertex is the earliest in the path from $l_{P_1}$ to $r_{P_1}$. If the color is adjacent to three or more different colors, then we can use the BFS tree argument to prove the paths induce a single cycle (i.e., some of the colors in the cycle will not be on the path from $l_{P_1}$ to $r_{P_1}$ in the BFS tree). Otherwise, if we have a discontiguous color that is adjacent to two or fewer different colors, then we have an unmerged pair of paths. This is not possible, because the input of the cycle-breaking procedure is a path system that is unwrapped, and therefore, merged.
%
% What was commented out is not true (the issue arises from the last part where we say that the path is unmerged; we can have a merged path with discontiguous colors!)
\end{proof} 

\subsubsection{Procedure to break special cycles}\label{app:break-cycles}

The cycle-breaking procedure takes as input a sequence of paths (the selected paths) that induce a special cycle $C$ containing the edge $e^\star$ and works on breaking $C$ by modifying the path system. As a consequence of the cycle being special, the only intersections that exist within the sequence of paths are between two contiguous paths or the first and the last path in the sequence of paths.

This procedure aims to reduce the number of paths that induce $C$ by looking for a \emph{source-target corner}. A source-target corner is a tree that is a subgraph of the path system graph and that is induced by two paths, such that updating the paths by swapping their targets and reconstructing them via the tree they induce makes it possible to delete one of the two updated paths without destroying $C$. Every time a source-target corner is found, we eliminate it by updating the paths as described earlier then deleting one of them and we end up with a smaller set of paths that induces the same cycle. For example, in Figure~\ref{fig:cycle_breaking_st_corner}a, three paths induce the special cycle. A source-target corner exists between $P_2$ and $P_3$, because swapping targets, reconstructing the paths and removing $P_3$ preserves the cycle~(Fig.~\ref{fig:cycle_breaking_st_corner}b), so $P_3$ is definitively removed from the selected paths, as it no longer contributes any edge to the special cycle. This process does not increase the total weight of the path system, since the same edges are used for the reconstruction.
Note that if the number of remaining selected paths is odd, then there must exist another source-target corner; otherwise, the number of source vertices will not be equal to the number of target vertices in the graph that the selected paths induce.

\begin{figure}[t]
    \includegraphics[]{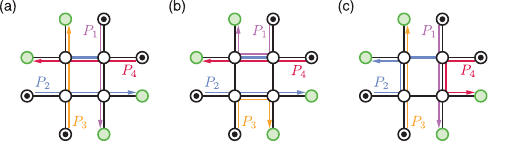}
    \caption{
    \label{fig:cycle_breaking_nontrivial}
    \textbf{Breaking a cycle in a path system induced on a reduced set of paths without source-target corners.} (a)~Example of path system with a special cycle formed by four paths. 
    (b)~Even reduced path system and (c)~odd reduced path system obtained after breaking the cycle. Because both reduced path systems have the same overall weight
    and the odd path system does not contain the edge $e^\star$ (thick blue line), the odd reduced path system gets picked.
    }
\end{figure}

%Non-trivial cycle breaking base case:
Eventually, we end up with a cycle with no source-target corners (and an even number of paths); if the cycle is made up of two paths, we use logic that is similar to the logic we presented for path merging (in the proof of Lemma~\ref{lemma:mps}). We attempt to reduce total path weight; if this fails, we attempt to isolate tokens; if this also fails, we reroute through the subpath that does not contain $e^\star$. 
If we end up with a cycle formed out of four or more paths~(see~Fig.~\ref{fig:cycle_breaking_nontrivial}), we refer to the remaining paths as a \emph{reduced set of paths}. Let $r$ be the cardinality of the reduced set of paths. We can generate two path systems induced on the reduced set of paths, which we call \emph{reduced path systems}, such that they each break the cycle. Let $\{P_1, P_2, \ldots, P_{r}\}$ be the reduced set of paths:

\begin{enumerate}
    \item \emph{Even reduced path system}: Match $v_{s_i}$ with $v_{t_i}$ if $i$ is even, $v_{s_i}$ with $v_{t_{((i+1) \mod r) + 1}}$ if $i$ is odd, $1 \leq i \leq r$. In the latter case, reconstruct path $P_i$ on the graph induced by the reduced set of paths without using its private edges.
    \item \emph{Odd reduced path system}: Match $v_{s_i}$ with $v_{t_i}$ if $i$ is odd, $v_{s_i}$ with $v_{t_{((i+1) \mod r) + 1}}$ if $i$ is even, $1 \leq i \leq r$. In the latter case, reconstruct path $P_i$ on the graph induced by the reduced set of paths without using its private edges.
\end{enumerate}

We remark that there is a single way to reconstruct the paths within each of the two path systems, as we are forbidding the inclusion of private edges within them, and if there were multiple reconstructions that would work, this would imply that the reduced set of paths induces more than one cycle.
The fact that both graphs induced by the reduced path systems are cycle-free follows from the fact that the reduced set of paths induces a single cycle, and from the existence of private edges for every path. Those private edges are no longer part of any path in each of the reduced path systems, and the number of distinct edges used did not increase; therefore the cycle has been broken.  
We would now like to update the reduced set of paths using one of the two reduced path systems we constructed.
Observe that the total weight of one of the two generated reduced path systems is less than or equal to that of the path system involving the selected paths. 
The observation follows from the fact that the total weight of the odd reduced path system and that of the even reduced path system add up to twice the total weight of the reduced set of paths. If the total weights of the two reduced path systems are different, we pick the reduced path system with the smallest total weight. Otherwise, we pick the reduced path system that does not include the edge $e^\star$ (Figure~\ref{fig:cycle_breaking_nontrivial}). After the selection is made, we update the paths in the path system accordingly.

It may be the case that the cycle-breaking procedure gives rise to pairs of paths that are unmerged, or paths that are wrapped; those are two possible consequences of the elimination of source-target corners. Therefore, every time a cycle is broken, we run the path merging and the path unwrapping procedures to reinstate the invariant (namely that the path system is a UPS), followed by repopulating the edge frequencies in the updated path system. The path system being unwrapped is what ensures the correctness of our work.

\subsubsection{Proof of termination of the cycle-breaking procedure}\label{app:termination}
Finally, we show that the cycle-breaking procedure terminates. Once the cycle-breaking procedure terminates, the path system graph is a forest, which admits an ordering of moves (as we show in the next section).

\begin{lemma}
\label{lemma:cycle_breaking_termination}
Given an $n$-vertex (positive) edge-weighted graph $G$ with $|E(G)| = m$ and $\sum_{e \in E(G)}{w(e)} = \mathcal{O}(n^{c})$ for some positive integer $c$, two sets $S,T \subseteq V(G)$, and a $T$-valid unwrapped path system $\mathcal{P}$, the cycle-breaking procedure terminates in time $\mathcal{O}(n^{c + 6}m)$.
\end{lemma}

\begin{proof}
The cycle-breaking procedure was designed in a way that ensures it terminates in polynomial time. If we limited it to arbitrarily detecting and breaking cycles, it would have been harder to prove its termination, let alone that it terminates in polynomial time. 
Breaking a single cycle (which takes time $\mathcal{O}(n^4)$, with the proof to follow) has one of three consequences:

\begin{enumerate}
    \item The decrease of the overall weight of the path system, a possible increase in the number of isolated tokens in the path system, and an indeterminate effect on the edge frequencies
    \item The increase of the overall number of isolated tokens in the path system, with an indeterminate effect on edge frequencies (no increase in the overall weight of the path system)
    \item The decrease of the frequency of the least frequent edge that is part of a cycle, with an indeterminate effect on the frequency of the other edges (no increase in the overall weight of the path system or decrease in the number of isolated tokens)
\end{enumerate}

The consequences and their hierarchy are given by construction of the algorithm. The first consequence can occur $\mathcal{O}(n^{c+1})$ times; it may also increase the number of isolated tokens (but not decrease it, as cycle-breaking reroutes paths through edges that are already in the path system). Irrespective of the first consequence, the second consequence can occur $\mathcal{O}(n)$ times, because there are $\mathcal{O}(n)$ tokens, and the third consequence can occur $\mathcal{O}(nm)$ times in a row, since the frequency of a single edge is $\mathcal{O}(n)$, and edges that are taken out of cycles are not brought back into cycles because of how cycle-breaking is designed. One can observe that changes in path system weight or changes in the number of isolated tokens affect edge frequencies, in the sense that there can be $\mathcal{O}(nm)$ occurrences of the third consequence in a row after each occurrence of either of the first two consequences. Analogously to the reasoning we employed for path merging, this means that the cycle-breaking procedure involves breaking $\mathcal{O}((n^{c+1} + n)nm) = \mathcal{O}(n^{c+2}m)$ cycles. 

It remains to show that breaking a single cycle can be done in time $\mathcal{O}(n^4)$. In the worst case, there are $\mathcal{O}(n)$ paths that induce the cycle in question. In this case, constructing the reachability instances is accomplished in time $\mathcal{O}(n^4)$, as we end up with $\mathcal{O}(n^2)$ instances overall, each of which consists of $\mathcal{O}(n)$ vertices, such that constructing the edges for each instance takes time $\mathcal{O}(n^2)$ (provided that we precompute which pairs of paths intersect, which is executed once and in time $\mathcal{O}(n^2m)$ prior to the construction of the reachability instances). Every single BFS call on each reachability instance takes time $\mathcal{O}(n^2)$, which is the number of edges in a single path intersection graph, so all BFS calls combined take time $\mathcal{O}(n^4)$. We now account for the running time of source-target corner eliminations. In the worst case, we have $\mathcal{O}(n)$ source-target corner eliminations, each of which is executed in time $\mathcal{O}(n + m)$ (via BFS). Finally, it should be easy to see that the base cases run in time $\mathcal{O}(nm)$, as they either involve merging a pair of paths ($\mathcal{O}(m)$), or constructing two path systems and computing their total weight ($\mathcal{O}(nm)$). 
Given all the above, it follows that the cycle-breaking procedure terminates in time $\mathcal{O}(n^{c + 6}m)$. 
\end{proof}

\subsection{Step 4 -- Order path system}\label{app:ops}

The fourth step of the aro algorithm is to compute an \emph{ordered path system} (OPS) whose moves can be executed with the guarantee that each token moves at most once. 
  
\begin{theorem}[Ordered path system]\label{thm:ops}
Given an $n$-vertex (positive) edge-weighted graph $G$ with $|E(G)| = m$ and $\sum_{e \in E(G)}{w(e)} = \mathcal{O}(n^{c})$ for some positive integer $c$, two sets $S,T \subseteq V(G)$, and a $T$-valid path system $\mathcal{P}$, we can compute, in time $\mathcal{O}(n^{c + 6}m + n^3)$, a valid ordered path system $\mathcal{P'}$ such that $w(\mathcal{P'}) \leq w(\mathcal{P})$. Moreover, the number of distinct edges used in $\mathcal{P'}$ is at most the number of distinct edges used in $\mathcal{P}$.
\end{theorem}

\begin{proof}
Given $\mathcal{P}$, we apply Theorem~\ref{thm:cps} to compute a cycle-free path system $\mathcal{P''}$ in time $\mathcal{O}(n^{c + 6}m)$. Recall that $\mathcal{P''}$ is both merged and unwrapped and $G[\mathcal{P''}]$ is a forest. Hence, we can apply the algorithm of~\cite{Calinescu2007} on each tree of the forest. It was shown in~\cite{Calinescu2007} that given a tree with $n$ vertices, and with the number of source vertices $|S|$ in the tree being equal to the number of target vertices $|T|$ in the tree, there is a $\mathcal{O}(n)$-time algorithm, which we call the \emph{exact tree solver}, that performs the optimal (minimum) number of moves to transform $S$ to $T$ while moving each token at most once. It is easy to see that the non-trivial trees (i.e., the trees not containing a single source and target vertex) induced by the path system each have an equal number of source vertices and target vertices; if not, there will exist a token in some tree that does not need to move, which implies that the path system can be modified in a way that decreases its total weight (or increases the number of isolated/fixed in place tokens). So we can assume, without loss of generality, that each tree has an equal number of source vertices and target vertices. We keep track of the moves produced by the algorithm for each tree and add them in order to obtain $\mathcal{P'}$; the order among the different trees is irrelevant. The usage of the exact tree solver guarantees that $w(\mathcal{P'}) \leq w(\mathcal{P})$ and that the number of distinct edges used in $\mathcal{P'}$ is at most the number of distinct edges used in $\mathcal{P}$, as isolated tokens will not move (in fact, we might even increase the number of isolated tokens) and the set of edges used in $\mathcal{P'}$ is a subset of the edges used in $\mathcal{P}$. Moreover, the algorithm of~\cite{Calinescu2007} can be easily adapted to guarantee that the frequency of each edge, i.e., the number of times it is traversed in the path system, cannot increase since it is always traversed in one of its two possible directions. 
To see why $\mathcal{P'}$ is ordered, we can construct a directed graph $D$ where each path $P$ in $\mathcal{P'}$ corresponds to a node $v_{P}$ in $D$ and we add a directed edge from node $v_P$ to $v_{P'}$ whenever another path $P'$ in $\mathcal{P'}$ depends on $P$. We claim that $D$ is a directed acyclic graph. Suppose not. Then the existence of a cycle implies that either some pair of paths is unmerged, or some pair of paths is wrapped, or at least one token must move more than once, a contradiction in all cases. Hence, we can reconstruct the ordering of the moves by simply computing a topological ordering of $D$, which can be done in time $\mathcal{O}(|V(D)| + |E(D)|) = \mathcal{O}(n^2)$~\cite{kahn1962}. Note that constructing the graph $D$ takes time $\mathcal{O}(n^{3})$ in the worst case.
\end{proof}

\section{Obstruction solver subroutine}\label{sec:obstruction_solver}

The \emph{obstruction solver subroutine} seeks a sequence of moves associated with a valid path system. Our implementation of the obstruction solver subroutine is a variation of the subroutine devised by Călinescu \emph{et al.}~\cite{Calinescu2007}, the core difference being that our version can handle token surplus (i.e., $|S| > |T|$). Processing every path in the path system in an arbitrary order, the subroutine attempts to move each token from its source vertex to its target vertex. If an obstructing token is present on the path, the subroutine switches the target of the token that it is attempting to move with the target of the obstructing token, updates the path system with the previously-computed shortest paths between the newly updated pairs of source and target vertices, and then attempts to move the obstructing token. The recursive procedure terminates when all vertices in $T$ are occupied.

Formally, suppose that the token $\tau_i$ is on the source vertex $v_{s_i}$ of a path $P_i\in\mathcal{P}$ aiming towards its target vertex $v_{t_i}$.
If the move associated with $P_i$ is not obstructed, i.e., there is no other token on the path between $v_{s_i}$ and $v_{t_i}$, then $\tau_i$ is moved to $v_{t_i}$. Otherwise, there is some obstructing token that is closest to $v_{s_i}$, say $\tau_j$, on the path $P_i$. The obstruction solver subroutine finds the target vertex $v_{t_j}$ associated with token $\tau_j$ and then switches the target of $\tau_i$ with that of $\tau_j$. The solver updates the path system by choosing the shortest paths for the updated pairings of source and target vertices. The shortest paths are previously computed during the APSP subroutine and stored in memory, which requires $\mathcal{O}(n^2)$ space and does not present any scalability issues. After choosing the shortest paths, the solver recursively attempts to move the obstructing token. Because it started with a valid path system, the solver is guaranteed to return a valid sequence of unobstructed moves in polynomial time ($\mathcal{O}(n^2)$ time in the worst case).

The performance of the obstruction solver subroutine depends on the ordering of the paths on which moves are executed. Indeed, the ordering affects the number of displaced tokens, as well as the total number of transfer operations. Figure~\ref{fig:fig1a}a, an example on a path system that induces a cycle, shows an instance in which the execution of the ordering $(P_1, P_2, P_3)$ displaces every token once for a total of three moves, which is optimal, whereas the execution of the ordering $(P_2, P_3, P_1)$ requires five moves (given that the tokens are moved sequentially, one after the other), with tokens $\tau_2$ and $\tau_3$ being displaced twice. Our proposed ordering subroutine (Sec.~\ref{sec:ordering}) mitigates this excessive number of control operations by computing a partial ordering of the paths that guarantees that no token is displaced more than once.

The ordering of the paths in the execution of the obstruction solver subroutine also affects the number of tokens that are displaced. For example, Figure~\ref{fig:fig1a}b shows an instance in which three tokens are displaced for one ordering $(P_3, P_2, P_1)$, whereas two tokens are displaced for another ordering $(P_2, P_1, P_3)$. Indeed, in the latter ordering, $P_2$ is not obstructed, so it can be directly executed. Then, because $P_1$ is obstructed by token $\tau_3$, the target of $\tau_1$ becomes $v_{t_3}$ and the target of $\tau_3$ becomes $v_{t_1}=v_{s_3}$. The obstruction solver subroutine first attempts to move $\tau_3$, and, because it already occupies its target, there is nothing to be done. The solver then attempts to move $\tau_1$ to $v_{t_3}$. 
If the shortest path between $v_{s_1}$ and $v_{t_3}$ goes through $P_2$, then, because a move on this path is obstructed by $\tau_2$, the solver first moves $\tau_2$ to $v_{t_3}$ and then $\tau_1$ to $v_{t_2}$; two tokens have thus been displaced instead of three. Because token $\tau_3$ can be discarded from the path system, token $\tau_3$ can be isolated; the isolation subroutine~(App.~\ref{app:isolation}) seeks to find tokens that can be isolated and removes them from the path system to reduce unnecessary displacement operations.

\section{Batching subroutine}\label{sec:batching}
Because typical assignment-based reconfiguration algorithms do not take into account the number of transfer operations, the resulting control protocols might perform as many extraction and implantation operations as displacement operations, i.e., an EDI cycle for each elementary displacement operation. To reduce the number of transfer operations, we implement a \emph{batching subroutine} that seeks to simultaneously displace multiple atoms located on the same row or column of the grid graph within a single EDI cycle. Given that the output of the aro algorithm is a sequence of displacements, a batch is a maximal sequence of consecutive displacements, each of which displaces a distinct atom, such that the displacements’ sources are adjacent along a row/column, and the displacements’ targets are adjacent along the same row/column. If the reconfiguration system allows for simultaneous displacements within the same row or within the same column, then all atoms in a batch can be displaced at once. This simple batching routine could be further extended to achieve greater operational performance. The running time of the batching subroutine is no more than $\mathcal{O}(n^3)$. 

Although the resulting performance of our baseline algorithm is improved over a typical assignment-based algorithm that does not rely on the isolation and batching subroutines, our benchmarking analysis shows that a larger gain in operational performance is achieved by the aro algorithm, which further includes a rerouting subroutine and an ordering subroutine.

\end{appendix}

%\bibliography{main.bib}

%

\end{document}